\documentclass[11pt]{article}
\usepackage{graphicx} 
\usepackage{a4wide}
\usepackage{amsmath}
\usepackage{url}
\usepackage{epsfig}
\usepackage{psfrag}
\usepackage{amssymb}
\usepackage{wrapfig}
\usepackage{algorithmicx}
\usepackage{algpseudocode}
\usepackage{paralist}
\usepackage{xcolor}
\usepackage{xfrac}
\usepackage{comment}
\usepackage{multicol}
\usepackage[margin=1in]{geometry}
\usepackage{authblk}
\usepackage{tikz}
\usepackage{subcaption}
\usetikzlibrary{decorations.pathmorphing, shapes, arrows.meta, positioning, svg.path}

\newcommand{\floor}[1]{\left\lfloor #1 \right\rfloor}

\newcommand{\donotshow}[1]{}

\newcommand{\ignore}[1]{}
\newcommand{\pnegskipvar}[1]{\vspace*{-\baselineskip}\vspace*{-#1 \belowdisplayskip}\par}

\newcommand{\pnegskip}{\vspace*{-\baselineskip}\vspace*{-\belowdisplayskip}\par}

\newcommand{\assign}{\mathbin{\raisebox{0.05ex}{\mbox{\rm :}}\!\!=}}

\providecommand{\CC}{C\raisebox{.08ex}{\hbox{\tt ++}}}

\newcommand{\N}{\mathbb{N}}

\newcommand {\abs}[1] {| #1 |}

\newlength{\mysetspacing}
\providecommand{\sset}[1]{\{\hspace{\mysetspacing} #1 \hspace{\mysetspacing}\}}

\newcommand{\mbegin}{\{\ \ }
\newcommand{\mend}{\}}

\newlength{\mleftindent}
\newlength{\mindent}
\newlength{\mboxwidth}
\newcommand{\mincrement}{\addtolength{\mboxwidth}{-\mindent}}
\newcommand{\mdecrement}{\addtolength{\mboxwidth}{\mindent}}

\newlength{\preprogramskip}
\newlength{\postprogramskip}
\newlength{\mexpwidth}
\newlength{\mexpindent}
\newcommand{\indentafterkeyword}{\hspace*{0.5em}}

\newcommand{\mslifelse}[3]  
{\setlength{\mexpwidth}{\mboxwidth}%
\settowidth{\mexpindent}{{\bf if\indentafterkeyword}}%
\addtolength{\mexpwidth}{-\mexpindent}%
{\bf if\indentafterkeyword}\parbox[t]{\mexpwidth}{#1}\\
\mincrement \mbegin \parbox[t]{\mboxwidth}{#2 \mend} \mdecrement \\
{\bf else} \\
\mincrement \mbegin \parbox[t]{\mboxwidth}{#3}\\
\mend \mdecrement
}

{\vspace{-0.3em}\begin{itemize}
\setlength{\itemsep}{0.2\itemsep}%
\setlength{\parskip}{0.2\parskip}%
\setlength{\topsep}{0.0\topsep}%
}
{\end{itemize}}

{\begin{itemize}
\setlength{\itemsep}{0.2\itemsep}%
\setlength{\parskip}{0.2\parskip}%
\setlength{\topsep}{0.0\topsep}%
}
{\end{itemize}}

{\begin{itemize}
\setlength{\itemsep}{1.5\itemsep}%
\setlength{\parskip}{1.5\parskip}%
\setlength{\topsep}{0.0\topsep}%
}
{\end{itemize}}

{\begin{itemize}
\setlength{\itemsep}{1.5\itemsep}%
\setlength{\parskip}{1.5\parskip}%
\par \smallskip
}
{\end{itemize}}

{\begin{itemize}
\setlength{\itemsep}{1.5\itemsep}%
\setlength{\parskip}{1.5\parskip}%
\par \medskip
}
{\end{itemize}}

{\vspace{-0.3em}\begin{description}
\setlength{\itemsep}{0.2\itemsep}%
\setlength{\parskip}{0.2\parskip}%
\setlength{\topsep}{0.0\topsep}%
}
{\end{description}}

{\begin{description}
\setlength{\itemsep}{0.2\itemsep}%
\setlength{\parskip}{0.2\parskip}%
\setlength{\topsep}{0.0\topsep}%
}
{\end{description}}

{\vspace{-0.3em}\begin{enumerate}
\setlength{\itemsep}{0.2\itemsep}%
\setlength{\parskip}{0.2\parskip}%
\setlength{\topsep}{0.0\topsep}%
}
{\end{enumerate}}

{\begin{enumerate}
\setlength{\itemsep}{0.2\itemsep}%
\setlength{\parskip}{0.2\parskip}%
\setlength{\topsep}{0.0\topsep}%
}
{\end{enumerate}}

\newlength{\proofpostskipamount}\newlength{\proofpreskipamount}
\newenvironment{proof}%
               {\par\vspace{\proofpreskipamount}\noindent{\bf Proof:}\hspace{0.5em}}
               {\nopagebreak%
                \strut\nopagebreak%
    \hspace{\fill}\qed\par\vspace{\proofpostskipamount}\noindent}
\newenvironment{proofnoqed}%
               {\par\vspace{\proofpreskipamount}\noindent{\bf Proof:}\hspace{0.5em}}
               {\nopagebreak%
                \strut\nopagebreak%
                \hspace{\fill}\par\vspace{\proofpostskipamount}\noindent}

               {\par\vspace{0.5ex}\noindent{\bf Proof #1:}\hspace{0.5em}}%
               {\nopagebreak%
                \strut\nopagebreak%
                \hspace{\fill}\qed\par\medskip\noindent}

\newtheorem{lemma}{Lemma}

\providecommand{\qed}{\rule[-0.2ex]{0.3em}{1.4ex}}

\newlength{\mydefwidth}
\newlength{\mytextwidth}

\newcommand{\myurl}[1]{{\footnotesize \url{#1}}}

\usepackage{mathptmx}
\usepackage[scaled]{helvet}   
\usepackage{courier}

\usepackage[T1]{fontenc}
\usepackage[utf8]{inputenc}

\newcommand{\lcp}{\mathit{lcp}}  \newcommand{\Lcp}{\mathit{Lcp}} 
\newcommand{\lcpodd}{\mathit{lcp}_{\mathit{odd}}}
\newcommand{\Lcpodd}{\mathit{Lcp}_{\mathit{odd}}}
\newcommand{\mate}{\mathit{mate}}
\newcommand{\sap}{\mathit{sap}}  \newcommand{\saps}{\mathit{saps}} 
\newcommand{\hide}[1]{}
\newcommand{\htmladdnormallink}[2]{#1}

\begin{document}

\title{Gabow's $O(\sqrt{n}m)$ Maximum Cardinality Matching Algorithm, Revisited}
\author[1]{Kurt Mehlhorn}
\author[2]{Romina Nobahari}
\affil[1]{Max Planck Institute for Informatics, Saarbr\"ucken, Germany}
\affil[2]{EPFL, Lausanne, Switzerland}
\maketitle

\begin{abstract} We revisit Gabow's $O(\sqrt{n} m)$ maximum cardinality matching algorithm (The Weighted Matching Approach to Maximum Cardinality Matching, 
Fundamenta Informaticae, 2017). It adapts the weighted matching algorithm of Gabow and Tarjan~\cite{GT91} to maximum cardinality matching. Gabow's algorithm works iteratively. In each iteration, it constructs a maximal number of edge-disjoint shortest augmenting paths with respect to the current matching and augments them. It is well-known that $O(\sqrt{n})$ iterations suffice. Each iteration consists of three parts. In the first part, the length of a shortest augmenting path is computed. In the second part, an auxiliary graph $H$ is constructed with the property that shortest augmenting paths in $G$ correspond to augmenting paths in $H$. In the third part, a maximal set of edge-disjoint augmenting paths in $H$ is determined, and the paths are lifted to and augmented in $G$. We give a new algorithm for the first part, a new proof for the second part (necessary, since the first part changed), and correct an oversight in the third part. Gabow's algorithm for the first part is derived from Edmonds' primal-dual algorithm for weighted matching. Our approach avoids the use of linear programming duality. We believe that it is more direct and easier to teach. We have implemented the algorithm; the implementation is available at the companion webpage (\url{https://people.mpi-inf.mpg.de/~mehlhorn/CompanionPageGenMatchingImplementation.html}).
\end{abstract}

\tableofcontents

\section{Introduction}

The maximum matching problem is one of the basic problems in graph theory and graph algorithms. Given an undirected graph, the goal is to find a matching, i.e., a set of edges no two of which share an endpoint, of maximum cardinality. Edmonds gave a polynomial time algorithm as early as 1965~\cite{Edmonds:matching}. The running time of the algorithm was improved over time, culminating in the $O(nm\alpha(n))$ algorithm of Gabow~\cite{Gabow:edmonds} and the $O(nm)$ algorithm of Gabow and Tarjan~\cite{Gabow-Tarjan:union-find}. An implementation of the former algorithm is available in LEDA~\cite{LEDAsystem,LEDAbook}. Kececioglu and Pecqueur~\cite{Kececioglu:matching} give heuristic improvements that often lead to considerably smaller running times.  

Algorithms with running time $O(\sqrt{n} m)$ were given in~\cite{MV80,Vazirani94,Vazirani12,Vazirani20,Vazirani24,Goldberg-Karzanov,GT91,Gabow:GeneralMatching}.  Very recently (September 24, 2026), OpenAI~\cite{OAI} announced an algorithm with almost linear running time. \cite{Gabow:f-matching} generalizes~\cite{Gabow:GeneralMatching} to $f$-matchings. Mattingly and Ritchey~\cite{Mattingly-Ritchey} and Huang and Stein~\cite{Huang-Stein} discuss implementations of the Micali-Vazirani algorithm, and Ansaripour, Danaei, and Mehlhorn~\cite{GabowImplementation}  give an implementation of Gabow's algorithm. 

Gabow's $O(\sqrt{n}m)$ algorithm works iteratively, and so do the other $O(\sqrt{n}m)$ algorithms. It is based on Edmonds' matching algorithm~\cite{Edmonds:matching}. In each iteration, it constructs a maximal number of edge-disjoint shortest augmenting paths with respect to the current matching and augments them. It is well-known that $O(\sqrt{n})$ iterations suffice. Each iteration consists of three parts. The first part determines the length of a shortest augmenting path. The second part constructs an auxiliary graph $H$ with the property that shortest augmenting paths in $G$ correspond to augmenting paths in $H$. The third part determines a maximal set of edge-disjoint augmenting paths in $H$, lifts the paths to $G$, and finally augments them in $G$. Gabow's algorithm is derived from the algorithm for maximum weight matching in~\cite{GT91}. The latter algorithm also works iteratively, with each iteration having three parts. In the first part, the weight of a minimum weight augmenting path is determined. For the cardinality matching problem, the computation of a minimum weight augmenting path (= shortest augmenting path) considerably simplifies, as Gabow shows. The second and the third part of the algorithm in~\cite{Gabow:GeneralMatching} are as in~\cite{GT91}. In~\cite{Gabow:GeneralMatching}, the correctness proof for part three is more detailed than in~\cite{GT91}. 

\emph{We give a new algorithm for the first part}. Gabow's algorithm for the first part is derived from Edmonds' primal-dual algorithm for weighted matching. We do not use linear programming duality and give a purely combinatorial algorithm instead. We believe that our approach is more direct and will be easier to understand to readers not familiar with linear programming duality. 

In Section~\ref{part I},  we describe how to find a shortest augmenting path. This section is our main contribution. In Section~\ref{part II}, we describe the construction of $H$ and the search for a maximal set of augmenting paths and their lifting to $G$. The construction of $H$ is an adaptation of Gabow's approach to our changes in the first part. The search for a maximal set of augmenting paths and their lifting to $G$ is as in~\cite{Gabow:GeneralMatching} and~\cite{GT91}. We include the algorithm for the third part and its proof of correctness for completeness. We also correct an oversight in~\cite{Gabow:GeneralMatching}. In Section~\ref{Connection to Gabow}, we discuss the relationship to Gabow's approach for Part I. The implementation is available on the companion webpage: \url{https://people.mpi-inf.mpg.de/~mehlhorn/CompanionPageGenMatchingImplementation.html}

\section{The New Algorithm for Part I: Finding a Shortest Augmenting Path}\label{part I}

Let $G = (V,E)$ be the input graph, and let $M$ be the current matching. Initially, $M$ is empty. Some odd length cycles in $G$ are called \emph{blossoms}. The cycle is an alternating path of odd length with the first and the last edge non-matching. The common endpoint of the first and the last edge is called the \emph{base} of the blossom. In order for an odd length cycle to qualify for a blossom, there must be an even length alternating path starting at a free node and ending at the base. The path is called the \emph{stem} of the blossom. Stem and cycle are simple paths and vertex-disjoint except for the base. Figure~\ref{fig:blossom-example} gives an example.
Contraction of the blossom into the base yields a contracted graph $G/C$. It is well-known~\cite{Edmonds:matching} that there is an augmenting path in $G$ iff there is an augmenting path in $G/C$. Blossoms can be nested in the following sense. Assume there is a blossom in $G/C$ whose blossom involves the base of a previously contracted blossom. Contraction of the blossom then results in a blossom in which the previously contracted blossom is nested. A blossom is \emph{maximal} if it is not nested in a larger blossom. Figure~\ref{Example} illustrates nested blossoms.

To find a shortest augmenting path with respect to $M$, we grow a collection $S$ of search structures, contractions of subgraphs of $G$, one for each free vertex of $G$. Each search structure is an alternating tree, i.e., along any root to leaf path, the edges alternate between matching and non-matching, the edge incident to the root being non-matching. We use “node” for the vertices of the search structures and “vertex” for the vertices of $G$. We frequently refer to a search structure as a tree. The growth is structured into phases $0$, $1$, \ldots; we use $\Delta$ for the number of the current phase. In phase zero, we initialize $S$ with the free vertices. In later phases, we grow alternating trees rooted at the free nodes. Each node of $S$ is labeled \emph{even} or \emph{odd}, depending on the parity of its depth; roots have depth zero and are therefore even. The nodes of $S$ correspond to sets of vertices of $G$ with odd cardinality. Odd nodes represent single vertices, whereas even nodes may correspond to larger subsets and then have internal structure. Even nodes are contractions of blossoms, maybe nested. We refer to an even node of $S$ as a \emph{blossom}. Even nodes of cardinality one are trivial blossoms. Consider a tree with free vertex\footnote{The root of the tree is a blossom containing the free vertex.} $f$ and a vertex $v$ inside an even node of the tree. Then there will be an even-length alternating path, maybe several,\footnote{Recall that even nodes have structure. In the contracted graph, there is a unique path from the blossom containing $v$ to $f$. However, taking the inside structure of blossoms into account, there may be several paths. For example, if $v$ is contained in a non-trivial blossom, the path from $v$ to the base of the blossom is not unique.} in $S$ from $v$ to $f$. For this reason, the vertices inside a blossom are also called even. For the odd nodes of $S$, there will be an odd-length alternating path, maybe several, in $S$ from $v$ to $f$, but no even-length alternating path. Among these paths, we will single out a shortest and call it the \emph{canonical path} of the vertex. The concrete definition of canonical paths will be given in the next section.

For a matching edge, either both endpoints belong to $S$ or none does. Vertices not belonging to $S$ are \emph{unlabeled}. The unlabeled vertices come in pairs of vertices matched to each other. For a matched vertex $v$, let $\mate(v)$ denote its matching partner. When a vertex is added to $S$, it is added either as an even or as an odd node. We also say, it is born even or odd. A node born odd may become even at a later time, namely when it becomes part of a non-trivial blossom. A node born even stays even. 

We compute two values for each vertex of $G$. For a vertex $v$, 
$\lcp(v)$, the \emph{length of the canonical path}, is (intended to be) the length of a shortest even length alternating path in $S$ from a free node to $v$ (if any), and, for a node born odd, $\lcpodd(v)$ is (intended to be) the length of a shortest odd-length alternating path in $S$ from a free node to $v$. In Lemma~\ref{correctness}, we will prove that $\lcp(v)$ and $\lcpodd(v)$ actually have their intended meaning. The \emph{level of a node} is either $\lcp(v)$ or $\lcpodd(v)$ depending on whether it is born even or odd. Free nodes are born even and have level zero. 

\begin{figure}[t]
\centering

\begin{subfigure}[t]{0.22\textwidth}
\centering
\begin{tikzpicture}[
  every node/.style={circle,draw,inner sep=2pt},
  squig/.style={decorate,decoration={zigzag,segment length=4,amplitude=0.9},thick},
  scale=0.55,
  transform shape
]

\node (a) at (0,-8.5){$a$};
\node (b) at (0,-7){$b$};
\node (c) at (0,-5.5){$c$};
\node (d) at (-1,-3.5){$d$};
\node (g) at (1,-3.5){$g$};
\node (e) at (-1,-2){$e$};
\node (f) at (1,-2){$f$};
\node (h) at (-2.5,-5.5){$h$};
\node (i) at (-2.5,-3.5){$i$};

\draw (a)--(b);
\draw[squig] (b)--(c);
\draw (c)--(d);
\draw[squig] (d)--(e);
\draw (e)--(f);
\draw[squig] (g)--(f);
\draw (c)--(g);
\draw (b)--(h);
\draw[squig] (h)--(i);
\end{tikzpicture}
\label{fig:blossom-a}
\end{subfigure}
\hfill
\begin{subfigure}[t]{0.22\textwidth}
\centering
\begin{tikzpicture}[
  every node/.style={circle,draw,inner sep=2pt},
  squig/.style={decorate,decoration={zigzag,segment length=4,amplitude=0.9},thick},
  scale=0.55,
  transform shape
]

\node[draw=blue,text=blue] (a) at (0,-8.5){$a$};
\node[draw=blue,text=blue] (b) at (0,-7){$b$};
\node[draw=blue,text=blue] (c) at (0,-5.5){$c$};
\node[draw=blue,text=blue] (d) at (-1,-3.5){$d$};
\node[draw=blue,text=blue] (g) at (1,-3.5){$g$};
\node[draw=blue,text=blue] (e) at (-1,-2){$e$};
\node[draw=blue,text=blue] (f) at (1,-2){$f$};
\node (h) at (-2.5,-5.5){$h$};
\node (i) at (-2.5,-3.5){$i$};

\draw[blue] (a)--(b);
\draw[squig,blue] (b)--(c);
\draw[blue] (c)--(d);
\draw[squig,blue] (d)--(e);
\draw[blue] (e)--(f);
\draw[squig,blue] (g)--(f);
\draw[blue] (c)--(g);

\draw (b)--(h);
\draw[squig] (h)--(i);

\end{tikzpicture}
\label{fig:blossom-b}
\end{subfigure}
\hfill
\begin{subfigure}[t]{0.22\textwidth}
\centering
\begin{tikzpicture}[
  every node/.style={circle,draw,inner sep=2pt},
  squig/.style={decorate,decoration={zigzag,segment length=4,amplitude=0.9},thick},
  scale=0.55,
  transform shape
]

\node (a) at (0,-8.5){$a$};
\node (b) at (0,-7){$b$};
\node[draw=red,text=red] (c) at (0,-5.5){$c$};
\node[draw=red,text=red] (d) at (-1,-3.5){$d$};
\node[draw=red,text=red] (g) at (1,-3.5){$g$};
\node[draw=red,text=red] (e) at (-1,-2){$e$};
\node[draw=red,text=red] (f) at (1,-2){$f$};
\node (h) at (-2.5,-5.5){$h$};
\node (i) at (-2.5,-3.5){$i$};

\draw (a)--(b);
\draw[squig] (b)--(c);

\draw[red] (c)--(d);
\draw[squig,red] (d)--(e);
\draw[red] (e)--(f);
\draw[squig,red] (g)--(f);
\draw[red] (c)--(g);

\draw (b)--(h);
\draw[squig] (h)--(i);

\end{tikzpicture}
\label{fig:blossom-c}
\end{subfigure}
\hfill
\begin{subfigure}[t]{0.22\textwidth}
\centering
\begin{tikzpicture}[
  every node/.style={circle,draw,inner sep=2pt},
  squig/.style={decorate,decoration={zigzag,segment length=4,amplitude=0.9},thick},
  scale=0.55,
  transform shape
]

\node (a) at (0,-8.5){$a$};
\node (b) at (0,-7){$b$};
\node(c) at (0,-5.5){$c-g$};
\node (h) at (-2.5,-5.5){$h$};
\node (i) at (-2.5,-3.5){$i$};

\draw (a)--(b);
\draw[squig] (b)--(c);


\draw (b)--(h);
\draw[squig] (h)--(i);

\end{tikzpicture}
\label{fig:blossom-d}
\end{subfigure}

\caption{A graph containing a blossom. From left to right, the original graph, the blossom together with its stem in blue, the blossom with base $c$ in red, and the contracted graph $G/C$.}
\label{fig:blossom-example}

\end{figure}

\subsection{The Growth of the Search Structures}

We use two operations for growing the search structures: \emph{Growth-steps} and \emph{bridge-steps}. Growth-steps are only performed in even phases. Recall that we use $\Delta$ to number the phases. For even $\Delta$, we perform growth-steps at even nodes with $\lcp$-value $\Delta - 2$. For even and odd $\Delta$, we perform bridge-steps at non-matching edges $uv \in E$ with even endpoints and $\lcp$-sum $2\Delta - 2$, i.e., $\lcp(u) + \lcp(v) = 2\Delta - 2$. The choice of the parameters $\Delta - 2$ and $2\Delta - 2$ will become clear below. 

In the initialization step, we make each free vertex the root of a search structure. The $\lcp$-values of the free nodes are equal to zero. This concludes phase zero. 

\paragraph{Growth-Steps:} A growth-step adds two nodes to $S$. Recall that growth-steps only occur when $\Delta$ is even. The first growth-steps occur in phase two. Let $u$ be an even vertex with $\lcp(u) = \Delta - 2$, and let $e = uv$ be an edge with $v$ unlabeled. Then $e \not\in M$ and $\mate(v)$ is also unlabeled. We extend $S$ by making $v$ a child of the node containing $u$ and $\mate(v)$ a child of $v$ and set 
$\lcpodd(v) = \lcp(u) + 1$ and $\lcp(\mate(v)) = \lcp(u) + 2$. The canonical paths to $v$ and $\mate(v)$ are the canonical path to $u$ extended by the edge $uv$ and by the edges $uv$ and $v\mate(v)$, respectively.

\paragraph{Bridge-Steps:} We consider even-even non-matching edges $uv$ with $\lcp(u) + \lcp(v) = 2\Delta - 2$, where $u$ and $v$ do not belong to the same maximal blossom, i.e., belong to different maximal blossoms. Recall that a blossom is maximal if it is not nested in a larger blossom. The first bridge-step may occur in phase one; this will be the case if there is an edge connecting two free vertices. 
\begin{description}
    \item[$u$ and $v$ belong to distinct trees:] We have found an augmenting path of length $\lcp(u) + \lcp(v) + 1 = 2\Delta - 1$, and Part I ends. 

    \item[$u$ and $v$ belong to the same tree:]  We have found a blossom. Let $B_u$ and $B_v$ be the blossoms containing $u$ and $v$, respectively, and let $B$ be the lowest (= furthest from the root) common ancestor of $B_u$ and $B_v$ in $S$. If one of $B_u$ or $B_v$ is an ancestor of the other, then $B$ is equal to this node and hence even. Otherwise, $B_u$ and $B_v$ are in different subtrees with respect to $B$, and hence $B$ is an even node, as odd nodes have only a single child. Note that, the child of an odd node is a blossom containing the mate of the odd node. We add the edge $uv$ to $S$. The edge $uv$ together with the paths from $B_u$ and $B_v$ to $B$ forms a cycle. All odd nodes on both paths become even. 

    The \emph{newly formed blossom} comprises the following nodes: All nodes contained in any of the blossoms on both paths plus the odd nodes on both paths. Note that the even nodes in a blossom may be blossoms shrunken earlier. A blossom comprises an odd number of vertices of $G$. This is obvious for trivial blossoms. A non-trivial blossom is formed by an odd-length cycle, and any node of the cycle stands for an odd number of vertices of $G$; the sum of an odd number of odd numbers is odd. For all but one vertex of a blossom, the mate also belongs to the blossom. The single vertex that is not mated inside the blossom is the base of the blossom; it is either a free vertex or has its mate outside the blossom. We shrink the blossom into a supernode. The supernode has an even level in $S$, the level of its base.
    
Consider an odd node $z$ on the path from $B_v$ to $B$. The canonical path to $z$ consists of the canonical path to $u$ followed by the edge $uv$ followed by the reversal of a suffix of the canonical path to $v$. It is the suffix starting in $z$. Its length is: 
\[ \lcp(z) = \lcp(u) + 1 + \lcp(v) - \lcpodd(z), \]
see Figure~\ref{lcp for odd nodes}. For an odd node on the path from $B_u$ to $B$, the analogous statement holds. The canonical paths to all nodes in a blossom pass through the base of the blossom. The length of the suffix of the canonical path of $z$ starting in the base $b$ of the blossom is $\lcp(z) - \lcp(b)$. 
\end{description}

For future reference, we summarize the discussion in a Lemma. 

\begin{lemma}\label{structure of saps} Consider a blossom $B$ with base $b$ and a node $v$ of $B$. Then $b$ lies on the canonical path of $v$, and the suffix of that path that starts in $b$ and ends in $v$ has length $\lcp(v) - \lcp(b)$. \end{lemma}

\begin{figure}[t]
\begin{center}
\begin{tikzpicture}[
  every node/.style={circle, draw, inner sep=2pt},
  squig/.style={decorate, decoration={zigzag, segment length=4, amplitude=0.9}, thick},
  scale=0.7, transform shape
]

\node (a) at (0,-7){$r$};
\node (b) at (0,-5.5){$b$};
\node (c) at (-1,-4.5){};
\node(d) at  (+1,-4.5){};
\node (x) at (-1,-2.0){$u$};
\node (z) at (+1,-3.5){$z$};
\node (y) at (+1,-2.0){$v$};

\draw[dashdotted] (a) -- (b);
\draw (b) -- (c);
\draw (b) -- (d);
\draw[dashdotted] (c) -- (x);
\draw[dashdotted] (d) -- (z);
\draw[dashdotted] (z) -- (y);
\draw (x) -- (y);
\end{tikzpicture}
\end{center}
\caption{\label{lcp for odd nodes}
The vertex $z$ becomes even by the addition of the bridge $uv$. The path from $r$ to $u$ has length $\lcp(u)$, the path from $r$ to $v$ has length $\lcp(v)$, the path from $r$ to $z$ has length $\lcpodd(z)$, the path from $v$ to $z$ has length $\lcp(u) - \lcpodd(z)$; therefore, $\lcp(z) \assign \lcp(u) + 1 + \lcp(v) - \lcpodd(z)$. Moreover, the even length path from $b$ to $z$ has length $\lcp(z) - \lcp(b)$. Edges are drawn solid, and paths are drawn dash-dotted.}
\end{figure}

\begin{figure}[t]
\input{Figure1}
 \caption{\label{Example}  \textbf{The execution of part I on the example given in~\cite{Gabow:GeneralMatching}. All the graphs are as in~\cite{Gabow:GeneralMatching}.}: (a) Shows the input graph for the search. (b) At the end of phase 0, $S$ consists of the two free nodes $a$ and $l$. (c) At the end of phase 2, we have grown $S$ to level 2. (d) We grow out of $c$ using edge $cd$ and set $\lcpodd(d) = 3$ and $\lcp(e) = 4$. When we consider $ce$, we set $\lcpodd(e) = 3$. After the growth-step in phase 4, $S$ has grown to level 4. There is an even-even edge $ce$ with $\lcp(c) + \lcp(e) = 2 + 4 = 2\cdot 4 - 2$. (e) After the bridge-step in phase 4, $d$ is now even with $\lcp(d) = 4$. (f) In phase 5, we add the bridge $ik$ with $\lcp(i)+ \lcp(k) = 4 + 4 = 2\cdot 5 - 2$. The vertices $h$ and $j$ become even with $\lcp$-value equal to $6$. (g) In phase 6, we grow the tree and add node $r$ at level 5 and node $q$ at level 6. Then we process the bridge $pq$ with $\lcp(p) + \lcp(q) = 4 + 6 = 2\cdot6 - 2$, $r$ becomes even with $\lcp(r) = 6$. We also add the bridge $eh$ with $\lcp(e) + \lcp(h) = 4 + 6 = 2\cdot 6 - 2$. Vertices $b$ and $f$ become even with $\lcp$-value equal to $10$. (h) Finally, in phase 7, we add the bridge $fn$ with $\lcp(f) + \lcp(n) = 10 + 2 = 2\cdot 7 - 2$. We have found an augmenting path of length 13, and Part I ends. (j) The contracted graph: Vertices $a$ to $k$ are contracted into the red supernode, and vertices $p$ to $r$ into the green supernode. Edges $mf$ and $pf$ are discarded. (i) also shows the structure inside the supernodes.  }
\end{figure}

\paragraph{Synchronization between Growth-Steps and Bridge-Steps:}

The parameters $\Delta - 2$ for growth-steps and $2\Delta - 2$ for bridge-steps are chosen for proper synchronization of growth- and bridge-steps. We present the intuition now and defer the detailed justification to the next section. Recall that we want to find a shortest augmenting path. We find an augmenting path when we encounter a bridge connecting two different trees. This suggests processing bridge-steps in order of increasing $\lcp$-sum. Similarly, we should grow the trees breadth-first, i.e., perform growth-steps out of nodes in increasing order of their $\lcp$-values. When we grow out of a node at level $\Delta - 2$, we create new even nodes at level $\Delta$ that may form bridges with existing nodes at level $\Delta - 2$ (not smaller by Lemma~\ref{bridges created by growth steps}) and hence we should perform bridge-steps with combined value $2\Delta -2$ not before growth steps out of nodes at level $\Delta - 2$. When we process a bridge with $\lcp$-sum $2\Delta - 2$, we may create new even nodes at level $\Delta$ (not smaller by Lemma~\ref{new lcp-values}) and hence we may perform growth-steps out of nodes with $\lcp$-value $\Delta$ after processing bridge-steps with $\lcp$-sum $2\Delta - 2$. So we use the order:
For even $\Delta$, growth-steps out of nodes with $\lcp$-value $\Delta - 2$. Then bridge-steps with $\lcp$-sum $2\Delta - 2$. Then bridge-steps with $\lcp$-sum $2\Delta$. Increase $\Delta$ by two and repeat.

We summarize: In phase zero, we initialize the search structure with the free nodes. Each free node becomes the root of an alternating tree. In phase $\Delta$, $\Delta \ge 1$, the following additions to the search structure take place:
\begin{description}
\item[Delta is even, growth-steps:] We grow out of nodes $v$ with $\lcp(v) = \Delta - 2$. As long as there is a node $v$ with $\lcp(v) = \Delta - 2$ with an unlabeled neighbor, say $x$, we add $x$ and $\mate(x)$ to $S$, $x$ at level $\Delta - 1$ with parent $v$, and $\mate(x)$ at level $\Delta$ with parent $x$.
\item[Delta is even or odd, bridge-steps:] We add even-even edges $xy$ with $x$ and $y$ belonging to different maximal blossoms and $\lcp(x) + \lcp(y) = 2\Delta - 2$. This may make some odd nodes even. They have $\lcp$-value larger than $\Delta$ (Lemma~\ref{new lcp-values}). If $x$ and $y$ belong to different trees, we have found an augmenting path; if they belong to the same tree, we have found a blossom. 
\end{description}
Figure~\ref{Example} illustrates an execution of the algorithm. The pseudo-code is shown in Figure~\ref{pseudocode}.

\begin{figure}[t]
\begin{algorithmic}[t!]
  \Function{Part1}{} \Comment{Returns True if an augmenting path exists and False otherwise}
  \State Initialize search structures to empty structures;
  \For{($\Delta = 0$, $\Delta = \Delta +1$, $\Delta \le n+1$)}   \Comment{Phase $\Delta$}
  \If{$\Delta$ is even}
  \State Grow the search structures to level $\Delta$ by either adding all free nodes ($\Delta = 0$)
  \State or for $\Delta \ge 2$ by growing out of nodes $v$ with $\lcp(v) = \Delta - 2$ as follows:
  \For{all vertices $v$ with $\lcp(v) = \Delta - 2$ and all edges $vx$}
  \If {$x$ is unlabeled}\smallskip
   \State \parbox{0.7\textwidth}{Make $x$ a child of $v$ and $\mate(x)$ a child of $x$, label $x$ odd and $\mate(x)$ even, and set
  $\lcpodd(x) \leftarrow \Delta - 1$ and $\lcp(\mate(x)) \leftarrow \Delta$;}
  \Else
  \If {$x$ is labeled even and $\lcp(x) = \Delta$}
  \State{ $\lcpodd(x) = \Delta - 1$; } \Comment{only needed for the correctness proof}
  \EndIf 
  \EndIf
  \EndFor
  \EndIf
  \While{$\exists$ even-even edge $xy$ with $\lcp(x) + \lcp(y) = 2\Delta - 2$ connecting \\ \hspace{2.1cm} different maximal blossoms}
  \If{$x$ and $y$ belong to the same search structure} \Comment{blossom}
  \State \begin{minipage}[t]{0.7\textwidth}Add the edge to the search structure; let $b$ be the base of the blossom formed. Make all odd nodes on the paths from $x$ and $y$ to $b$ even. For any odd node $z$ in the newly formed blossom,  set $\lcp(z) = \lcp(x) + 1 + \lcp(y) - \lcpodd(z)$. \end{minipage}
  \Else \Comment{$\sap$, set up $H$}
  \State Construct the contracted graph $H$;
  \State \textbf{return} true;
  \EndIf
  \EndWhile
 \EndFor
 \State \textbf{return} false;
 \EndFunction
\end{algorithmic}
\caption{\label{pseudocode}Part I of the Matching Algorithm. Augmenting paths are only found for $\Delta \le n/2$. The remaining phases are needed for the correct construction of the witness of optimality; see Section~\ref{Implementation}.}
\end{figure}

\subsection{Properties of Part I and Proof of Correctness}
In this section, we establish structural properties of the search graph $S$ and its evolution during the execution of the algorithm. We show that the discovery of vertices and bridges follows a well-defined order governed by the canonical path lengths ($\lcp$). In particular, growth- and bridge-steps processed in phase $\Delta$ only generate events for phases greater than or equal to $\Delta$. Finally, we prove that the canonical paths in $S$ correspond to shortest alternating paths, which implies the correctness of the algorithm.

\begin{lemma}[Even-odd edges go up at most one level]\label{even-odd edges go up at most one level} Let $uv$ be an edge with $u$ even and $v$ odd. Then $\lcpodd(v) \le \lcp(u) + 1$. \end{lemma}
\begin{proof} Consider, when the edge $uv$ is scanned out of even node $u$. The largest $\lcpodd$-value assigned up to and including this point is $\lcp(u) + 1$. Thus $\lcpodd(v) \le \lcp(u) + 1$. \end{proof}

\begin{lemma}[Discover bridges in order]\label{even-even in order} Assume, we process a bridge $xy$ and form a blossom. Let $v$ be an odd node that lies on the path from $y$ to the base of the blossom. Let $uv$ be a non-matching edge connecting $v$ to an even node $u$; $u$ and $v$ may be in different trees.  Then $v$ becomes even and $\lcp(x) + \lcp(y)\le \lcp(v) + \lcp (u)$. \end{lemma}
\begin{proof} We have $\lcp(v) = \lcp(x) + 1 + \lcp(y) - \lcpodd(v)$, and hence $\lcp(v) + \lcpodd(v) = \lcp(x) + 1 + \lcp(y)$. Moreover, by Lemma~\ref{even-odd edges go up at most one level}, $\lcpodd(v) \le \lcp(u) + 1$, since $v$ is labeled odd at the latest when we scan the edge $uv$. Thus
  \[ \lcp(x) + 1 + \lcp(y) = \lcp(v) + \lcpodd(v) \le \lcp(v) + \lcp(u) + 1.  \]
  \pnegskipvar{1.6}
\end{proof}

\begin{lemma}[Bridge-steps generate growth-steps only for later phases] \label{new lcp-values} Let $xy$ be a bridge with $\lcp$-sum $\lcp(x) + \lcp(y) = 2 \Delta - 2$, and let $v$ be a node that becomes even by the addition of the bridge. Then $\lcp(v) \ge \Delta$. \end{lemma}
\begin{proof} Odd nodes are only added in growth-steps. In phase $\Delta$, all odd nodes have an $\lcpodd$-value of at most $\Delta - 1$. Therefore, 
  \[ \lcp(v) = \lcp(x) + 1 + \lcp(y) - \lcpodd(v) = 2\Delta - 1 - \lcpodd(v) \ge 2 \Delta - 1 - (\Delta - 1) = \Delta.\]  \pnegskipvar{1.6}
\end{proof}

The \emph{skewness} of a bridge $xy$ is the magnitude (absolute value) of the difference of the $\lcp$-values of its endpoints, i.e. $\abs{\lcp(x) - \lcp(y)}$. We call a bridge \emph{horizontal} if its skewness is zero or two. Consider a horizontal bridge with $\lcp$-sum $2\Delta -2$. For odd $\Delta$, the $\lcp$-values of both endpoints are then $\Delta - 1$, for even $\Delta$, the $\lcp$-values of the endpoints are then $\Delta - 2$ and $\Delta$. Bridge-steps make odd nodes even. If both endpoints have $\lcp$-value $\Delta - 1$, the newly even nodes have $\lcp$-value at least $\Delta + 1$ and hence either do not participate in a bridge event in phase $\Delta$ or in a non-horizontal bridge event. If the endpoints have $\lcp$-values $\Delta - 2$ and $\Delta$, the newly even nodes have $\lcp$-value at least $\Delta$. Newly even nodes with $\lcp$-value $\Delta$ may be part of horizontal bridges in the same $\Delta$-phase. Figure~\ref{horizontal bridges}(b) shows an example.

\begin{figure}[t]
\begin{minipage}[t]{0.45\textwidth}
\begin{center}
\begin{tikzpicture}[
  every node/.style={circle, draw, inner sep=2pt},
  squig/.style={decorate, decoration={zigzag, segment length=4, amplitude=0.9}, thick},
  scale=0.7, transform shape
]
\node (h) at (-0.8,-1.2) {h};
\node (a) at (-0.8,0) {a};
\node (b) at (1,0) {b};
\node (f) at (5,-4.8){f};
\node (c) at (1,-1.2) {c};
\node (g) at (5,-6){g};
\node (d) at (3,-2.4) {d};
\node (e) at (3,-3.6){e};

\draw[dashed] (a) -- (b);

\draw[dashed] (d) -- (c);
\draw[dashed] (e) -- (f);
\draw[squig] (b) -- (c);
\draw[squig] (h) -- (a);
\draw[squig] (f) -- (g);
\draw[squig] (e) -- (d);

\node[draw=none, inner sep=0, anchor=west] at (-3.5,0) {$\Delta - 1$};
\node[draw=none, inner sep=0, anchor=west] at (-3.5,-1.2) {$\Delta - 2$};
\node[draw=none, inner sep=0, anchor=west] at (-3.5,-2.4) {$\Delta - 3$};
\node[draw=none, inner sep=0, anchor=west] at (-3.5,-3.6) {$\Delta - 4$};
\node[draw=none, inner sep=0, anchor=west] at (-3.5,-4.8) {$\Delta - 5$};
\node[draw=none, inner sep=0, anchor=west] at (-3.5,-6) {$\Delta - 6$};
\end{tikzpicture}

\centering{(a)}
\end{center}

\end{minipage}\hfill
\begin{minipage}[t]{0.45\textwidth}
\centering
\begin{tikzpicture}[
  every node/.style={circle, draw, inner sep=1.2pt},
  squig/.style={decorate, decoration={zigzag, segment length=4, amplitude=0.9}, thick},
  scale=0.7, transform shape
]

\node[draw=none, inner sep=0, anchor=west] at (-3.5,0) {$\Delta$};
\node[draw=none, inner sep=0, anchor=west] at (-3.5,-1.2) {$\Delta - 1$};
\node[draw=none, inner sep=0, anchor=west] at (-3.5,-2.4) {$\Delta - 2$};
\node[draw=none, inner sep=0, anchor=west] at (-3.5,-3.6) {$\Delta - 3$};
\node[draw=none, inner sep=0, anchor=west] at (-3.5,-4.8) {$\Delta - 4$};
\node[draw=none, inner sep=0, anchor=west] at (-3.5,-6) {$\Delta - 5$};

\node (b) at (-1.2,-3.6) {b};
\node (c) at (0,-3.6) {c};
\node (d) at (1.2,-3.6){d};
\node (e) at (-1.2,-2.4){e};
\node (f) at (0,-2.4){f};
\node (g) at (1.2,-2.4){g};
\node (h) at (0,-1.2) {h};
\node (i) at (0, 0){i};
\node (j) at (2.4, -2.4) {j};
\node (k) at (2.4, -3.6) {k}; 
\node (l) at (3.6, -4.8){l};
\node (m) at (3.6, -6){m};

\draw(f) -- (h);
\draw[squig] (b) -- (e);
\draw[squig] (c) -- (f);
\draw[squig] (d) -- (g);
\draw[squig] (h) -- (i);
\draw[squig] (j) -- (k); 
\draw[squig] (l) -- (m); 
\draw[dashed] (e) -- (i);
\draw[dashed] (g) -- (h);
\draw[dashed] (h) -- (j);
\draw[dashed] (k) -- (l);
\end{tikzpicture}
\medskip

\centering{(b)}  
\end{minipage}
\caption{\label{horizontal bridges}\label{cascading bridges}\textbf{Cascading bridges}:
(a) $\Delta$ is assumed to be odd. At the beginning of phase $\Delta$ , $a$, $b$, $d$, and $f$ are even, and $c$, $e$, and $g$ are odd. Nodes $a$ and $b$ are guaranteed to have their correct $\lcp$-values by the first part of Lemma~\ref{correctness}. We add the bridge $ab$ with $\lcp$-sum $2\Delta - 2$; its skew is zero. As a consequence, $c$ becomes even with $\lcp$-value $\Delta + 1$ and the edge $cd$ becomes an even-even edge with $\lcp$-sum $2\Delta - 2$ and a skewness of 4. We add $cd$, $e$ becomes even with the $\lcp$-value $\Delta + 3$, and $ef$ becomes an even-even edge with $\lcp$-sum $2\Delta - 2$ and skew 8. Note that the edges $ab$, $cd$, and $ef$ have increasing skew.  \\ \protect 
(b) $\Delta$ is assumed to be even. In the growth-step in phase $\Delta$ we have added node $h$ at level $\Delta - 1$ and node $i$ at level $\Delta$. The edge $ei$ is a bridge with $\lcp$-sum $2\Delta - 2$ and skew 2.  The addition of $ei$ makes $h$ even with $\lcp(h) = \Delta$ and creates the bridges $hg$ and $hj$ with skew 2. Node $k$ becomes even with $\lcp(k) = \Delta + 2$. Then $kl$ has two even endpoints, $\lcp$-sum $2\Delta - 2$, and skew 6. So we add $kl$. 
}\end{figure}

\begin{lemma}[Non-horizontal bridges: Skewness increases]\label{skewness increases} Let $z$ be a node that becomes even because of the non-horizontal bridge event $xy$. Then $\lcp(z) > \max(\lcp(x),\lcp(y))$ and hence if $z$ becomes part of a bridge for the same phase, the bridge has larger skewness. \end{lemma}
\begin{proof} Assume w.l.o.g.~$\lcp(x) < \lcp(y)$. Since $xy$ is a non-horizontal bridge, we have $\lcp(x) + 4 \le \lcp(y)$. Together with the constraint $\lcp(x) + \lcp(y) = 2\Delta - 2$, it follows that $\lcp(x) \le \Delta - 3$ and $\lcp(y) \ge \Delta + 1$. Since nodes that are born in growth-steps, have $\lcp$-value at most $\Delta$, $y$ was born odd. 

  Let $b$ be the base of the blossom formed by $xy$. For odd nodes $z$ on the path from $x$ to $b$, we have $\lcp(z) \ge \lcp(y) + 2$ . Consider next an odd node $z$ on the path from $y$ to $b$. Since $y$ was born odd, we have $\lcpodd(y) \le \lcp(x) + 1$ (Lemma \ref{even-odd edges go up at most one level}), and since $z$ is a proper ancestor of $y$, we have $\lcpodd(z) \le \lcpodd(y) - 2$ and hence $\lcpodd(z) \le \lcp(x) - 1$. Thus 
\[ \lcp(z) = \lcp(x) + 1 + \lcp(y) - \lcpodd(z) > \lcp(y).\] \pnegskip
\end{proof}

\begin{lemma}[Growth-steps in phase $\Delta$ generate bridge-steps for phases $\ge \Delta$]\label{bridges created by growth steps} Let $x$ be an even node and consider a growth-step in phase $\Delta$ to unlabeled node $y$. The step adds $y$ and $\mate(y)$ to the search structure; $y$ becomes odd and $\mate(y)$ becomes even. Let $z$ be another even node in the search structure, and assume there is an edge connecting $z$ to $\mate(y)$. Then $\lcp(z) \ge \lcp(x)$. In particular, the $\lcp$-sum of the bridge with endpoints $z$ and $\mate(y)$ is at least $2\Delta - 2$. \end{lemma}
\begin{proof} Since we grow out of $x$ in phase $\Delta$, $\lcp(x) = \Delta - 2$. For the sake of a contradiction, assume $\lcp(z) < \lcp(x)$. Then we grew out of $z$ before growing out of $x$. At this point, $\mate(y)$ and $y$ were unlabeled, and we would have added $\mate(y)$ as a child of $z$. So $\lcp(z) \ge \lcp(x)$.
Moreover, $\lcp(\mate(y)) = \lcp(x) + 2$ and hence $\lcp(\mate(y)) + \lcp(z) \ge \lcp(x) + 2 + \lcp(x) = 2\Delta - 2$.  \end{proof}

At this point, we know that the growth- and bridge-steps are well-ordered. Growth steps do not generate bridge-steps for an earlier phase; bridge-steps do not generate bridge-steps for an earlier phase and generate growth-steps only for a later phase. We next show that canonical paths are indeed shortest alternating paths.\smallskip

\newcommand{\SD}{S_\Delta}

For a node $v$, let $\Lcp(v)$ ($\Lcpodd(v)$) be the length of a shortest even-length (odd-length) alternating path in $G$ starting from a free vertex and ending in $v$; if no such path exists, the value is defined as $\infty$. In contrast, $\lcp(v)$ and $\lcpodd(v)$ are the corresponding values \emph{maintained by the algorithm}.

\begin{lemma}[Correctness]\label{correctness} For even $\Delta$, at the end of phase $\Delta$:  $S$ contains exactly the vertices that can be reached from a free vertex by an alternating path of length at most $\Delta$. Furthermore, $\lcp(v)$ is defined if $\Lcp(v) \le \Delta$, and $\lcpodd(v)$ is defined if $\Lcpodd(v) < \Lcp(v)\le \Delta$.

For any $\Delta$: Any non-matching edge $uv$ added in a bridge step satisfies $\Lcp(u) + \Lcp(v) = 2\Delta - 2$, and $\lcp(u) = \Lcp(u)$ and $\lcp(v) = \Lcp(v)$. If no augmenting path is found in the phase, all edges $uv$ with $\Lcp(u) + \Lcp(v) = 2\Delta - 2$ are added. Moreover, at any point in the execution, $\lcp(v) = \Lcp(v)$ whenever $\lcp(v)$ is defined, and $\lcpodd(v) = \Lcpodd(v)$ whenever $\lcpodd(v)$ is defined. 

The canonical path of an even vertex is a shortest alternating path from a free vertex ending in the vertex. When the algorithm terminates, it has found a $\sap$. 
\end{lemma}
\begin{proof} We use induction on $\Delta$. At the end of phase zero, $S$ consists precisely of the free nodes. They have $\Lcp$- and $\lcp$-value equal to zero. So the claim is true at the end of phase zero. In phase one, we add the non-matching edges connecting two free vertices. If there is such an edge, phase one is the last. Consider next any phase $\Delta$ with $\Delta \ge 2$. 

Assume first that $\Delta$ is even. We first show that at the end of phase $\Delta$ the search structures consist precisely of all vertices $v$ with $\Lcp(v) \le \Delta$ or $\Lcpodd(v) \le \Delta - 1$. 

Let $v$ be any node with $\Lcp(v) = \Delta$, let $p$ be a shortest even length alternating path ending in $v$, let $u$ be the node with $\Lcp(u) = \Delta - 2$ on $p$, and let $x$ and $v$ be the two last nodes on $p$.  Then $v = \mate(x)$, the edge $ux$ is non-matching, and $\Lcpodd(x) = \Delta - 1$.  We will show that $x$ and $v$ receive correct $\lcp$-values. By induction hypothesis, $S$ contains a canonical path of length $\Lcp(u)$, say $q$, ending in $u$. Also $\lcp(u) = \Lcp(u)$, and $q$ extended by the edges $ux$ and $xv$ is an even length alternating path starting in a free node, ending in $v$, and having length $\Delta$. We now distinguish cases according to the label of $v$ at the beginning of phase $\Delta$. 

If $v$ is unlabeled, $x$ is unlabeled, and both nodes are added in a growth-step in phase $\Delta$ at the latest when the edge $ux$ is considered for growth out of $u$. If the growth-step is along a non-matching edge incident to $x$ (this edge may be different from $ux$, but in any case, the growth is out of a node with $\lcp$-value equal to $\lcp(u)$), we set $\lcpodd(x) = \lcp(u) + 1 = \Lcp(u) + 1 = \Lcpodd(u)$ and $\lcp(v) = \lcp(u) + 2 = \Lcp(u) + 2 = \Lcp(v)$. If the growth-step is along a non-matching edge incident to $v$, we obtain $\lcpodd(v) = \Delta - 1 = \Lcpodd(v)$ and $\lcp(x) = \Delta = \Lcp(x)$, and the edge $ux$ becomes a bridge with $\lcp$-sum $2 \Delta - 2$. When this bridge is processed later in the phase, $v$ becomes even with $\lcp(v) = \Delta = \Lcp(v)$. In either case, the canonical paths to $x$ and $v$ are correctly defined. 

If $v$ is already labeled even, $v$ was added to $S$ in an earlier phase and hence $\min(\Lcp(v), \Lcpodd(v)) < \Delta - 2$. Together with $\Lcp(v) = \Delta$, we obtain $\Lcpodd(v) < \Delta - 2$. 

If $v$ is already labeled odd, $\lcpodd(v) \le \Delta - 3$ (growth-steps in earlier phases do not assign value $\Delta - 1$ or larger), and $x$ is labeled even with $\lcp(x) \le \Delta - 2$. Then $ux$ is an even-even edge with $\lcp$-sum at most $2\Delta - 4$, and the induction hypothesis ensures that it was considered for a bridge-step in an earlier phase. At that time, it was either added to $S$ or discarded because its endpoints already belonged to the same blossom. In the former case, when the bridge $ux$ was added, $v$ was labeled even with $\lcp(v) \le \Delta$. In the latter case, $x$ and therefore $v = \mate(x)$ were already labeled even. Thus, in either case, 
$v$ was already labeled even at the beginning of the phase, implying that this case cannot occur. 

\smallskip

From now on the parity of $\Delta$ is arbitrary. We need to show that all non-matching edges $uv$ added in the bridge steps of the phase  satisfy $\Lcp(u) + \Lcp(v) = 2\Delta - 2$ and $\lcp(u) = \Lcp(u)$ and $\lcp(v) = \Lcp(v)$. Moreover, if no augmenting path is found in the phase, all edges $uv$ with $\Lcp(u) + \Lcp(v) = 2\Delta - 2$ will have their endpoints in the same blossom. We use induction on the skew. 

We have already established that the $\lcp$- and $\lcpodd$-values less than or equal to $\Delta$ are correct. So, if $\max(\Lcp(u),\Lcp(v)) \le \Delta$, the edge will have  its endpoints in the same blossom. We have now established that all horizontal bridges will have their endpoints in the same blossom,  or a horizontal bridge with $\lcp$-sum $2\Delta - 2$ between different structures is found. 

Consider next a non-matching edge $uv$ with skew larger than two, and assume that all edges with smaller skew are added. We may assume $\Lcp(u) > \Lcp(v)$. Then $\Lcp(v) < \Delta - 2$ and hence $\lcp(v) = \Lcp(v)$ and $v$ is labeled even. For $u$, we have $\Lcp(u) > \Delta$. Thus, $u$ was born odd, $\Lcpodd(u) < \Delta$, and $\lcpodd(u) = \Lcpodd(u)$. Consider an alternating path of length $\Lcp(u)$ to $u$, and let $x - y - u$ be the last three vertices on this path. Then $xy$ is non-matching, and $yu$ is matching.  If $xy$ has a $\Lcp$-sum smaller than $2\Delta - 2$, it was added by the induction hypothesis. If it has $\Lcp$-sum $2\Delta - 2$, we have $\Lcp(x) = \Lcp(u) - 2$ and $\Lcp(y) = \Lcp(v) + 2 \le \Delta - 1$. Thus, the skew of $xy$ is four less than the skew of $uv$ and hence is non-negative and has a smaller absolute value. Thus, by induction hypothesis, $xy$ was added to $S$, and $\lcp(x) = \Lcp(x)$ and $\lcp(y) = \Lcp(y)$. In either case, $u$ was made even with $\lcp(u) = \lcp(x) + 2 = \Lcp(x) + 2 = \Lcp(u)$. Thus, $\lcp(u) + \lcp(v) = \Lcp(u) + \Lcp(v) = 2\Delta - 2$ and the algorithm adds $uv$. If the bridge is between different structures, phase one terminates and has determined a $\sap$. If the bridge closes a blossom, we continue. 
\end{proof}\smallskip

\subsection{The Priority Queue, the Maximum Value of Delta, and the Certificate of Optimality}\label{Priority Queue}

In phase $\Delta$, we grow out of even vertices with $\lcp$-value $\Delta - 2$ and we process even-even edges running between different blossoms and having $\lcp$-sum $2 \Delta - 2$. To schedule the events in their proper order, we maintain a priority queue. A simple queue, organized as an array of buckets, suffices. This is as in Gabow's algorithm. In bucket $\Delta$, we keep all events that need to be handled in phase $\Delta$, i.e., the even nodes with $\lcp$-value $\Delta - 2$ and the even-even edges with $\lcp$-sum $2\Delta - 2$. Lemmas~\ref{bridges created by growth steps}, \ref{even-even in order}, and~\ref{new lcp-values} guarantee that an event in phase  $\Delta$ does not generate an event for an earlier phase. 

How many buckets do we need? This depends on whether we only want to compute a maximum matching or also a certificate of optimality. In the former case, buckets $0$ to $\floor{n/2}$ suffice. Note that an augmenting path is found, when a bridge $xy$ connecting two even vertices in different trees is explored. Then $\lcp(x) + \lcp(y) + 1 \le n - 1$, since an augmenting path has length at most $n - 1$, and hence $\lcp(x) + \lcp(y) \le n - 2$. Thus, the bridge will be found at the latest in phase $\floor{n/2}$. Thus, if one is only interested in computing a maximum matching, the for-loop can be stopped once $2\Delta > n$. 

A certificate of optimality (\cite{LEDAbook}) can be given in the form of a function $\ell: V \rightarrow \N_{\ge 0}$. The vertex labeling $\ell$ has the property that for each edge $uv$ of $G$ either one endpoint is labeled $1$ or both endpoints are labeled with the same integer greater than one. Let $M$ be any matching, let $M_1$ be the edges in $M$ with at least one endpoint labeled $1$ and let $M_i$, $i \ge 2$, be the edges in $M$ with both endpoints labeled $i$. Let $n_i$, $i \ge 0$, be the number of nodes labeled $i$. Then $M = M_1 \cup \bigcup_{i \ge 2} M_i$, $\abs{M_1} \le n_1$, and $\abs{M_i} \le \floor{n_i/2}$, and hence $\abs{M} \le n_1 + \sum_{i \ge 0} \floor{n_i/2}$. \emph{Equality proves the optimality of $M$}. 

The labeling $\ell$ is readily constructed. However, this requires running phase I to completion, i.e., until the priority queue is empty. The maximum $\lcp$-value of any vertex is no more than $n-1$ and the maximum $\lcp$-sum of any edge is no more than $2n - 2$. So, buckets $0$ to $n - 1$ suffice, and once $\Delta$ exceeds $n - 1$, all edges are processed. If no augmenting path exists, all even-even edges must run inside a blossom. So, for every maximal non-trivial blossom, one labels all vertices of the blossom with the same integer greater than or equal to two; different integers are used for different blossoms. All even vertices not contained in a non-trivial blossom are labeled zero, and all odd nodes are labeled one. There might also be some edges that are not added to the search structures. These edges are matching. Their endpoints are only connected to odd vertices in the search structures. If there are matching edges entirely outside the search structures, handle them as follows. If there is only a single such edge, label one endpoint 0 and the other 1, and the witness equality is satisfied. If there are multiple such edges, say $n_k$, label one endpoint of one such edge 1 and assign all other endpoints a common new integer not used for any blossom. Note that there are $2n_k-1$ other endpoints. The witness inequality is satisfied since $1 + \floor{(2n_k - 1)/2} = n_k$. 

Figure~\ref{example against early stopping} shows an example that a certificate of optimality cannot yet be constructed after phase $\floor{n/2}$.

\begin{figure}[t]
    \begin{center}
\begin{tikzpicture}[
  every node/.style={circle, draw, inner sep=2pt},
  squig/.style={decorate, decoration={zigzag, segment length=4, amplitude=0.9}, thick},
  scale=0.75, transform shape
]
\node (a) at (0,0) {a};
\node (b) at (2,0) {b};
\node (c) at (4,0) {c};
\node (d) at (6,0){d};
\node (e) at (8,0) {e};
\node (f) at (10,1){f};
\node (g) at (10,-1) {g};

\draw[solid] (a) -- (b);
\draw[squig] (b) -- (c);
\draw[solid] (c) -- (d);
\draw[squig] (d) -- (e);
\draw[solid] (e) -- (f);
\draw[solid] (e) -- (g);
\draw[squig] (f) -- (g);

\end{tikzpicture}
\end{center}
\caption{\label{example against early stopping}The graph has seven vertices, and a maximum matching is shown. A certificate of optimality gives label one to vertices $b$ and $d$, label two to vertices $e$, $f$, and $g$, and label zero to vertices $a$ and $c$. After phase $\floor{7/2}$, the search structure consists of the path from $a$ to $c$, and the two matching edges $de$ and $fg$ still belong to the reservoir of unexplored edges. The labeling algorithm would give label one to $b$ and to one of the vertices in $\sset{d,e,f,g}$ and would give label two to the other vertices in this set. If $d$ is labeled two, the edge $cd$ connects labels zero and two and violates the witness property. After phase seven, all vertices belong to the search structure, and the correct labeling is computed. }
\end{figure}

\section{Parts II and III: From One to a Maximal Set of Augmenting Paths}\label{part II}
In parts II and III, we construct an auxiliary graph $H$, determine a maximal set of edge-disjoint augmenting paths in $H$, and lift them to a maximal set of shortest augmenting paths ($\sap$s) in $G$. We reuse Gabow's realization of these parts, which in turn uses the realization in~\cite{GT91}. The definition of $H$ and the proof that $H$ has the desired property need to be adapted, since Part I is now realized differently. We also slightly expanded the correctness proof for part III and corrected an oversight in Gabow's description. 


\newcommand{\dbase}{\mathit{dbase}} \newcommand{\base}{\mathit{base}}
\newcommand{\uh}{\mathit{u_H}} \newcommand{\vh}{\mathit{v_H}}
 
\subsection{Part II: Construction of H}

Let us refer to the phase, in which we find an augmenting path, as the \emph{breakthrough phase}. At breakthrough, we have determined one $\sap$ but our goal is to construct a graph $H$ containing all $\sap$s. The growth of the search structures is non-deterministic in the sense that, within each phase, the order in which vertices and edges are added to $S$ is not uniquely defined. However, the state at the end of a phase is unique in the sense that, independent of the order, the outermost blossoms comprise the same sets of vertices. A blossom is outermost at the end of a phase, if it is not contained in another blossom. We keep two union-find data structures $\base$ and $\dbase$ (delayed base) for the blossoms. During a phase, we perform unions on $\base$ and construct a log of the union-operations, but leave $\dbase$ unchanged. At the end of a phase (except for the last phase (= breakthrough phase), where we find an augmenting path), we use the log to perform the unions also on $\dbase$. In this way, we know the blossoms at the beginning of the phase.\footnote{Gabow maintains the nesting structure of the blossoms as a forest. He reverts to the beginning of the breakthrough phase by deleting the top level of this forest. Our approach saves some lines of code.} We also keep a log of the updates of $\lcp$-values during a phase and, in the breakthrough phase, use this log to revert the $\lcp$-values to their values at the beginning of the phase. We form $H$ from $G$ by contracting the blossoms existing before the breakthrough phase (these are the blocks of $\dbase$) and keeping only edges between them, namely the edges that could have been added to the search structures in some execution. We refer to the blossoms before the breakthrough phase as the maximal blossoms. 

\begin{description}
\item[Nodes:] The nodes of $H$ are the maximal blossoms, i.e., the blocks of $\dbase$. We represent each maximal blossom by its base. 
\item[Edges:] Consider any edge $uv$ of $E$, and let $\uh = \dbase(u)$ and $\vh = \dbase(v)$ be the corresponding nodes of $H$.
  \begin{itemize}
  \item If $(\uh,\vh)$ is a self-loop or both nodes are odd, the edge is discarded. 
  \item even-even: If $u$ and $v$ are even, $\uh \not= \vh$, and the bridge-condition $\lcp(u) + \lcp(v) = 2\Delta - 2$ is satisfied, the edge $(\uh,\vh)$ is added.  
  \item even-odd: Assume $u$ is even and $v$ is odd. If $(u,v)$ is non-matching and $\lcpodd(v) = \lcp(u) + 1$, the edge $(\uh,\vh)$ is added. If $(u,v)$ is matching, and $\lcp(u) = \lcpodd(v) + 1$, the edge $(\uh,\vh)$ is added. So non-matching edges go from even to (one larger) odd layer, and matching edges go from odd to one larger even layer. This is exactly as in growth-steps. 
  \item unlabeled-unlabeled: Matching edges are added, non-matching edges are not added. 
  \item even-unlabeled: If the even node has $\lcp$-value $\Delta - 2$, the edge is added. Note that the growth-steps in phase $\Delta$ would add these edges. These are the only even-unlabeled edges added. 
    \end{itemize}
  \end{description}

\begin{lemma}\label{Maximal blossoms and saps} Let $v$ be an odd node that becomes even in phase $\Delta$, and assume phase $\Delta$ is not the last. Let $B$ be the maximal blossom containing $v$ at the end of phase $\Delta$. Then there is no even node $u$ outside $B$ and connected to $v$ with $\lcpodd(v) = \lcp(u) + 1$. Conversely, any even neighbor $u$ of $v$ and outside $B$ satisfies $\lcp(u) \ge \lcpodd(v) + 1$. 

Assume that the breakthrough occurs in phase $\Delta$, and let $B$ be a maximal blossom at the end of phase $\Delta - 1$. Then any sap intersecting $B$ must pass through the base of $B$, i.e., the base is either one of the endpoints of the sap or the sap uses the matching edge incident to the base. 
\end{lemma}
\begin{proof} Assume otherwise, i.e., $v$ is connected to an even node $u$ outside $B$ with $\lcpodd(v) = \lcp(u) + 1$. Since $v$ became even in phase $\Delta$, $\lcp(v) = 2 \Delta - 1 - \lcpodd(v)$ and hence $\lcp(v) = 2 \Delta - 1 - \lcp(u) - 1$, i.e., $\lcp(u) + \lcp(v) = 2 \Delta - 2$ and $uv$ would be considered for addition in phase $\Delta$. Since phase $\Delta$ is not the last, the edge does not complete an augmenting path. Therefore, when $uv$ is considered for addition, either $u$ and $v$ already belong to the same blossom or do so after the addition of $uv$. So, $u$ and $v$ belong to the same maximal blossom at the end of phase $\Delta$, a contradiction. 

We have $\lcpodd(v) \le \lcp(u) + 1$, since $v$ must be odd after the growth-steps out of $u$. Together with $\lcpodd(v) \not= \lcp(u) + 1$, we obtain $\lcpodd(v) \le \lcp(u) - 1$. 

Since the breakthrough occurs in phase $\Delta$, saps have length $2 \Delta - 1$. Let $B$ be a maximal blossom at the end of phase $\Delta - 1$, and assume that there is a sap $P$ that intersects $B$ but does not pass through the base of $B$. 
We can write $P = P_1 e P_2$, where $e = z_1 z_2$ is an edge connecting two distinct search structures. The intersection of $P$ with $B$ is either part of $P_1$ or part of $P_2$, say it is part of $P_1$. 

$P_1$ enters $B$ on a non-matching edge $uv$ with $u$ even and outside $B$. Let $s_1 = \mate(v)$, $t_1$, $s_2$, $t_2$, $s_3$, \ldots, $t_{\ell - 1}$, $s_\ell = z_1$ be the vertices on the subpath of $P_1$ from $\mate(v)$ to $z_1$. Then $\lcp(s_1) \le \lcp(u)$ and $\lcp(s_i) \le \lcp(s_{i-1}) + 2$ for $i \ge 1$, and hence $\lcp(z_1) \le \abs{P_1} - 2$. Furthermore, $\lcp(z_2) \le \abs{P_2}$. Replacing $P_1$ by the canonical path to $z_1$ and $P_2$ by the canonical path to $z_2$ shortens the augmenting path and maintains simplicity of the path, a contradiction. 
\end{proof}

The lemma justifies that blossoms can be contracted at the end of any phase that is not the last. 
The following lemma is Corollary 3.3 in \cite{Gabow:GeneralMatching}. Since we changed part I, we need to reprove it.

\begin{lemma}\label{H and G} A set of edges $P$ forms an augmenting path in $H$ if and only if it is the image of a $\sap$ $Q$ in $G$. \end{lemma}
\begin{proof} For a path $p$ and a maximal blossom $B$, let $\gamma(p,B)$ be the part of $p$ inside $B$. Consider any $\sap$ $q$ in $G$ and any maximal blossom $B$ with $\gamma(q,B) \not= \emptyset$. By Lemma~\ref{Maximal blossoms and saps}, $\gamma(q,B)$ is an even-length alternating path connecting the base of $B$ with some other vertex of $B$. So, contraction of all such $\gamma(q,.)$ yields an augmenting path in $H$. Note that the edge of $p$ (if any) incident to the base of the blossom and outside the blossom, is matching, and the edge of $p$ incident to the other vertex  and outside the blossom is non-matching.  

  Conversely, consider an augmenting path $p$ in $H$ and any node of $H$ on $p$. If it represents a trivial blossom, the incident edges in $H$ correspond to edges in $G$. If it represents a maximal blossom $B$, one of the incident edges in $H$ is non-matching and the other (if it exists) is matching. Let $v$ be the endpoint in $B$ of the pre-image of the non-matching edge. Then there is an even-length alternating path connecting $v$ to the base $b$ of $B$ (Lemma~\ref{structure of saps}). The base of $B$ is free in $G$ if no matching edge is incident to the contraction of $B$ in $G$. If there is a matching edge incident to the contraction of $B$ in $p$, its pre-image must be the matching edge incident to the base of $B$ in $G$. Thus, we can lift $p$ to a augmenting path $q$ in $G$. 
  
  It remains to show that $q$ is a $\sap$. Let $uv$ be the edge of $q$ that was added last to $S$, say, in phase $\Delta$. Then $\lcp(u) + \lcp(v) = 2\Delta - 2$, and the path has length $2\Delta - 1$. Since no bridge between different trees was added in an earlier phase, it is a $\sap$. \end{proof}

Note that the second part of the proof of Lemma~\ref{H and G} shows how to lift augmenting paths from $H$ to $G$.

\newcommand{\eventime}{\mathit{eventime}}
\newcommand{\rep}{\mathit{rep}}
\newcommand{\PG}{\mathit{PG}}\newcommand{\free}{\mathit{free}}\newcommand{\baseH}{\mathit{baseH}} \newcommand{\bh}{\mathit{bh}}\newcommand{\eh}{\mathit{e_H}}

\subsection{Part III: Construction and Augmentation of a Maximal Set of Saps}\label{phase2}\label{Phase2}

This is as in~\cite{Gabow:GeneralMatching} with two changes. First, we expanded some arguments. Second, we corrected an oversight in Lemma A.2 in \cite{Gabow:GeneralMatching}; see Lemma~\ref{related or even-odd}. 

Recall the strategy for bipartite graphs. We first construct a layered network. The $i$-th layer contains all nodes that can be reached from a free node by an alternating path of length $i$; layer zero consists of the free nodes. We stop the construction once we reach a layer containing, again, a free node. The layered network can be constructed using \emph{breadth-first search}. The construction of the graph $H$ corresponds to the construction of the layered network; augmenting paths in $H$ are in one-to-one correspondence to saps in $G$. 

Then we construct a maximal set of edge-disjoint augmenting paths in the layered network. We explore the layered network from a free node using \emph{depth-first search}. When we reach a free node and hence have found an augmenting path, the path corresponds to the recursion stack and hence is readily extracted. We delete the path and all its incident edges from the graph simply by tracing back the recursion and declaring all nodes on the path finished. Moreover, when we retract from a node, we delete the node from further consideration, as we can be sure that no free node can be reached through the node. DFS satisfies two crucial properties: First, in the case of a breakthrough, the recursion stack contains the augmenting path, and second, in the case of a retraction from a node, the node can be ignored from then on. 

\emph{How can we extend this reasoning to $H$?} Again, we use depth-first search. DFS on an undirected graph partitions the set of edges into tree edges and non-tree edges. Non-tree edges connect two nodes that are related to each other, but are not connected by a tree edge. One has the choice of traversing a non-tree edge either as a forward edge, i.e., from ancestor to descendant, or as a backward edge, i.e., from descendant to ancestor. 
Blossoms are closed by forward/backward edges connecting two even nodes in the same tree. For blossom-closing edges, we have the choice whether to explore them as forward or backward edges. We will see that forward is the appropriate choice. Consider an edge $xy$ closing a blossom with $b(x)$ being an ancestor of $b(y)$. {When the DFS reaches $b(y)$, the nodes between $b(x)$ and $b(y)$ are only partially explored, i.e., for even nodes, the recursive call is still active, and some odd nodes may still become even. However, when the DFS returns to $b(x)$, the recursive calls for the even descendants have been made, and the blossoms contained in the subtree rooted at $b(x)$ have been formed. Therefore, closing blossoms on forward edges is the appropriate choice.}  

{When the DFS returns to $b(x)$, the surface path, i.e., the path from $b(x)$ to $b(y)$ outside blossoms, has the form 
$b(x)$, $u_1$, $v_1$, $u_2$, ${v_2}$, \ldots, $u_k$, $b(y)$ with $u_1$, $u_2$, \ldots $u_k$ being odd, and the parent of $u_i$, $1 \le i \le k$, being a node in the blossom with base $v_{i-1}$; here, $v_0 = b(x)$.} See Figure~\ref{blossom step}. We now explore the edge $b(x)b(y)$ (induced by $xy$) as a forward edge and unite the nodes $u_1$ to $u_k$ and the blossoms with bases $b(x)$, $v_1$ to $v_{k-1}$ and $b(y)$ into a blossom with base $b(x)$. As a result, the nodes $u_1$ to $u_k$ become even, and we need to make recursive calls for them. In what order? To keep the spirit of depth-first search, we should put $u_k$ to $u_1$ on the recursion stack (or an equivalent) with $u_k$ pushed first and $u_1$ pushed last and start to explore edges out of $u_1$ (the previously odd node closest to $b(x)$). Note that this guarantees that in the case of a breakthrough, the recursion stack contains the augmenting path, and that in the case of an unsuccessful call, the node can be ignored in the future. If the calls were made in a different order, the preceding sentence  would not be true. For example, if we were to search from $u_k$ first and be successful, the removal of the augmenting path would remove the justification for the nodes $u_1$ to $u_{k-1}$ to be even.

\newcommand{\findapset}{\mathit{find\_ap\_set}}
\newcommand{\findap}{\mathit{find\_ap}}
 \newcommand{\CP}{\mathit{CP}}

How can we distinguish forward and backward edges? To this end, we record for each node the time when it becomes even, i.e., we maintain a counter $t$, which we increment whenever a node becomes even and which we use to define $\eventime(v)$ for even vertices $v$. Then an edge $xy$ induces a  forward edge if $\eventime((b(x)) < \eventime(b(y))$. 

{The complete algorithm for phase 3 is shown in Figure~\ref{phase 3}. It defines procedures $\findapset$ and $\findap$ that construct a search structure $S$ and a collection $\CP$ of edge-disjoint augmenting paths. The top-level procedure $\findapset$ iterates over all free vertices and calls $\findap(f)$ for every free vertex $f$ that is not an endpoint of an augmenting path yet. $\findap$ is a recursive procedure that attempts to find an augmenting path. A successful search adds an augmenting path to $\CP$, an unsuccessful search leaves $\CP$ unchanged. 

The call $\findap(f)$ grows a search structure rooted at $f$. 
Consider a recursive call $\findap(x)$. Then $x$ is even. We scan all non-matching edges $xy$ incident to $x$.

If $y$ does not belong to $S$ yet and is free, an augmenting path has been found. The path is added to $\CP$ and all currently executing recursive calls of $\findap$ are terminated. If $y$ does not belong to $S$ yet and is matched, say to $y'$, the edges $xy$ and $yy'$ are added to the search structure and $\findap(y')$ is called. 

If $y$ already belongs to $S$ and $b(y)$ is an ancestor of $b(x)$, a blossom has been found, but we do \emph{nothing}, as blossoms are closed on forward edges.

If $y$ already belongs to $S$ and $b(y)$ is a proper descendant of $b(x)$ in $S$,  a blossom has been found. All odd vertices on the path from $b(x)$ to $b(y)$, say $u_1$ to $u_k$ with $u_1$ closest to $b(x)$ become even, all blossoms on the path are combined into a single blossom with base $b(x)$ and recursive calls $\findap(u_1)$, $\findap(u_2)$, \ldots, $\findap(u_k)$ are made in this order; see Figure~\ref{blossom step}.

If a top-level call $\findap(f)$ does not succeed in finding an augmenting path, it adds all nodes reachable from $f$ by an alternating path and not reached in an earlier search to $S$ and the call terminates regularly. Some nodes reachable from $f$ by an alternating path may already belong to $S$ when $\findap(f)$ is called. Such nodes belong to a search structure grown earlier (in a successful or unsuccessful search for an augmenting path). 

$\findap$ defines a canonical path $P(x)$ (as a sequence of nodes) for each even vertex $x$. For a free node, this is the trivial path consisting solely of the free node.  In a growth step, $P(y') = P(x)y y'$, and in a blossom step, $P(u_i) = P(x) P(y)_{u_i}^{\mathit{rev}}$ for $1 \le i \le k$. Here $P(y)_{u_i}$ is the suffix of the canonical path to $y$ starting in $u_i$ and $\mathit{rev}$ denotes reversal. 
}

\begin{figure}[bth]
\begin{algorithmic}[t!]
     \Function{$\findapset$}{} 
  \State Initialize $S$ to an empty graph and $\CP$ to an empty set of paths;
  \For{each free vertex $f$}   
  \If{$f \not\in V(\CP)$}
  \State add $f$ to $S$ as the root of a new search tree;
  \State $\findap(f)$\;
  \EndIf
  \EndFor
  \EndFunction
  \medskip

\Function{$\findap$}{$x$}   \Comment{$x$ is even}
\For{each non-matching edge $xy$}      \Comment{scan $xy$ from $x$}
\If {$y \not\in V(S)$}
  \If{$y$ is free}    \Comment{$y$ completes an augmenting path}
   \State add $xy$ to $S$ and add the path $P(x)y$ to $\CP$;
   \State terminate every currently executing recursive call to $\findap$;
   \Else  \Comment{grow step}
   \State add $xy$ and $y y'$ to $S$, where $yy'$ is a matching edge;
   \State $\findap(y')$\;
  \EndIf
\Else
\If{$b(y)$ is an even proper descendant of $b(x)$ in $S$}   \Comment{blossom step}
  \State \Comment{equivalent test: $b(y)$ became even strictly after $b(x)$};\smallskip
  \State \begin{minipage}[t]{0.82\textwidth}let $u_i$, $i = 1 \ldots k$, be the odd vertices on the path from $b(x)$ to $b(y)$, ordered such that $u_1$ is closest to $b(x)$.\end{minipage}\smallskip
  \State \begin{minipage}[t]{0.82\textwidth} combine all blossoms on the path from $b(x)$ to $b(y)$ into a blossom with base $b(x)$;\end{minipage}\smallskip
  \For {$i \leftarrow 1$ to $k$}
  \State $\findap(u_i)$\;
  \EndFor
  \EndIf
  \EndIf
  \EndFor
 \EndFunction
\end{algorithmic}
\caption{\label{phase 3} Phase III of the Matching Algorithm. Reprinted from~\cite{Gabow:GeneralMatching} with notation adapted. }
\end{figure}

We view the search structure as an ordered forest. {Blossoms are contracted to supernodes. The trees are ordered according to the time of their construction. For even nodes (= contracted blossoms), the children are ordered according to the time of addition to the search structure. In our figures, time flows from left to right, i.e., the child added first is the leftmost child, and the child added last is the rightmost child.  So, a  growth step from $x$ adding nodes $y$ and $y'$ adds $y$ as the new rightmost child of $B_x$.} For a blossom step induced by the edge $xy$, where $b(y)$ is a descendant of $b(x)$, let $B_{v_0}$, $ u_1$, $B_{v_1}$, $u_2$, $B_{v_2}$, \ldots, $u_k$, $B_{v_k}$  be the surface path from $v_0 =b(x)$ to $v_k = b(y)$. The nodes $u_1$ to $u_k$ become even by the blossom step and the nodes $v_0$ to $v_k$ were already even before the blossom step. For $v_i$, $0 \le i < k$, let $\ell_i$ and $r_i$ be the ordered list of children left and right of $v_{i+1}$, respectively, and let $c$ be the ordered list of children of $v_k$. The blossom step merges the nodes on the path into a blossom with children list $\ell_0 \ell_1 \ldots \ell_{k-1} c r_{k-1} \ldots r_1 r_0$. {Two nodes are related if they belong to the same tree and one is an ancestor of the other. For two unrelated nodes, one is left of the other if it either belongs to an earlier tree or belongs to the same tree and is a descendant of an earlier child of their least common ancestor than the other. }

{After contraction of the blossoms, the search structure is a tree. The active vertex is contained in a node on the rightmost path of the contracted search structure. Clearly, when we start the growth of a new tree, the root of the tree is rightmost. Consider next the inspection of an edge $xy$, where $x$ is even. By induction hypothesis, $x$ is contained in a node on the rightmost path. If $y$ is unlabeled and free, the search ends, and the augmenting path is obtained from the rightmost path by expansion of blossoms. If $y$ is unlabeled and matched, say to $y'$, we add $y$ as a new rightmost child of $b(x)$. If $y$ is labeled and $b(y)$ is an even proper descendant of $b(x)$, we form a blossom as described above. The newly formed blossom becomes the new $b(x)$. Since $b(x)$ was on the rightmost path before the formation of the blossom and $b(x)$ also represents the newly formed blossom, $b(x)$ stays on the rightmost path. }

A node $x$ is \emph{completely scanned} if $\findap(x)$ has been called and was not terminated prematurely. In particular, all non-matching edges incident to $x$ have been scanned. 

\begin{lemma} At termination: An even node that does not belong to $\CP$ is completely scanned. \end{lemma}
\begin{proof} Let $v$ be an even node that does not belong to $\CP$. Since $v$ is even, $\findap(v)$ has been called. Since $v \not\in \CP$, the call has not terminated prematurely, and hence $v$ is completely scanned. \end{proof}

\begin{lemma}[Extension of Lemma A.2 in \cite{Gabow:GeneralMatching}]\label{related or even-odd} At any point in the algorithm, let $xy$ be an edge with $x$ even and $y \in S$. Then either 
\begin{enumerate}
\item $b(x)$ and $b(y)$ are related, or 
\item $y$ is odd and to the left of $x$, or 
\item $y$ belongs to $\CP$ and is to the left of $x$, or 
\item $x$ belongs to $\CP$ and $y$ belongs to a tree whose growth started after $x$ was added to $\CP$. 
\end{enumerate} See Figure~\ref{Illustration of Lemmas}. \cite{Gabow:GeneralMatching} formulates: At any point of the algorithm, consider an edge $rs$ where $r$ is outer, i.e.~even, and $s \in S$. Either $s$ is inner and left of $r$ in $\cal{S}^{-}$, or $b(r)$ and $b(s)$ are related in $\cal{S}^{-}$. End of quote. The third case of $s$ belonging to $\CP$ (it could be even) and the fourth case of $r$ belonging to $\CP$ and $s$ belonging to a tree whose growth started after $r$ was added to $\CP$ are not mentioned. \end{lemma} 
\begin{proofnoqed} We use induction on the number of operations executed by the algorithm. 

\begin{description}
    \item[Growth Steps:] A growth step adds two new nodes, say $z$ and $z'$,  to the search structure, $z$ being odd and $z'$ being even. After the addition, the new nodes are the final two nodes of the rightmost path. Consider any edge $z's$. Since $z'$ is rightmost, $b(s)$ is either related to $z'$ or $s$ is to the left of $z'$. In the latter case, $s$ cannot be even and not belong to $\CP$, as then it would be completely scanned and $z'$ would already be part of the search structure (at the latest, $z'$ would be added as a child of $s$). If $s$ belongs to $\CP$, it does not belong to the current tree. If $s$ does not belong to $\CP$, it is odd. In either case, we have established the claim with $x = z'$ and $y = s$. 
    
    Consider any edge $zs$. The same reasoning applies. Since $z$ is rightmost, $b(s)$ is either related to $z$ or $s$ is left of $z$. In the latter case, $s$ cannot be even and not belong to $\CP$, as then it would be completely scanned and $z$ would already be part of the search structure (at the latest, $z$ would be added as a child of $s$). This establishes the claim with $x = s$ and $y = z$. 
    
    \item[Blossom Steps:] Consider a blossom step induced by the edge $vw$ with $b(w)$ being a descendant of $b(v)$. Then $b(v)$ is on the rightmost path. Consider any edge $xy$ with $x$ even after the blossom step and $y \in S$. If $b(x)$ and $b(y)$ are related after the blossom step, we are done. Otherwise, $b(x)$ and $b(y)$ were not related before the blossom step (a blossom step does not destroy any ancestor-relationship), and $x$ was either odd or even before the blossom step. 
    
    If $x$ was even before the blossom step, either $y$ was to the left of $x$ before the blossom step (and $y$ was either odd or belonged to $\CP$) or $x$ belonged to $\CP$. In the latter case, $x$ does not belong to the current tree. If $y$ also does not belong to the current tree, the claim holds by induction hypothesis. If $y$ belongs to the current tree, it belongs to a later tree than $x$. In the former case, $y$ cannot have moved to the right by the blossom step. 
    
    If $x$ was odd before the blossom step, $x$ is part of the newly formed blossom and therefore is contained in a node on the rightmost path after the blossom step. Thus, $b(y)$ is either related to $b(x)$ after the blossom step or to the left of $x$. If $y$ is to the left and is even, it must belong to $\CP$, as otherwise, we would have had an even-odd edge incident to a node outside $\CP$ going to the right before the blossom step. \hfill \qed
\end{description} 
\end{proofnoqed} 

The next three lemmas are as in~\cite{Gabow:GeneralMatching}. This might be surprising given that we reformulated the preceding Lemma. However, the goal of these lemmas is to show that Phase 3 finds a maximal set of augmenting paths, and therefore the two missing cases play no role.

\begin{lemma}[Lemma A.3 in \cite{Gabow:GeneralMatching}]\label{same blossom one} At any point in the algorithm, let $t$ be an odd vertex, whose even mate $t'$ is completely scanned, and let $s$ be an odd descendant of $t$. After a blossom step that makes $s$ even, $s$ and $t$ belong to the same blossom. See Figure~\ref{Illustration of Lemmas}.
\end{lemma}
\begin{proof} Let $P$ be the path from $t$ to $s$. We use induction on the length of $P$. All even nodes on $P$ are completely scanned. This holds since $t'$ is completely scanned. We will show that increasingly larger initial segments of $P$ are added to the blossom containing $t$. 

Among the odd vertices on $P$, let $u$ be the first to become even in a blossom step, $u = t$ is possible. If there are several vertices that become even in the same blossom step, take the shallowest one (= the one closest to $s$). Let $u'$ be the mate of $u$. Let that blossom step be triggered by the edge $xy$, where $b(x)$ is an ancestor of $b(y)$. Then $u$ lies on the path from $b(x)$ to $b(y)$. Since $b(x)$ is not completely scanned, $b(x)$ does not lie on $P$ and hence is a proper ancestor of $t$. After the blossom step $t$, $u$ and $u'$ belong to the same blossom. 

If $u = s$, we are done. Otherwise, let $v$ be the odd vertex that follows $u$ on $P$, and let $v'$ be its mate. Then $v'$ is completely scanned, and hence we can apply the induction hypothesis to $v$ and $s$. After the blossom step that makes $s$ even, $v$ and $s$ belong to the same blossom. Since $v$ was odd before the blossom step, it belongs to the same blossom as $u$ after the blossom step, and hence to the same blossom as $t$. 
\end{proof}

\begin{figure}[t]
\begin{subfigure}{0.3\textwidth}
\begin{center}
\usetikzlibrary{calc,decorations.pathmorphing}

\begin{tikzpicture}[ scale=0.6, transform shape,
    vertex/.style={circle,draw,fill=white,inner sep=1.6pt},
    base/.style={circle,draw,fill=white,inner sep=1.6pt},
    blossom/.style={dashed,thick},
    squig/.style={
        decorate,
        decoration={zigzag,segment length=3pt,amplitude=0.8pt},
        thick
    },
    lab/.style={font=\small}
]

\def\r{0.9}

\coordinate (C1) at (0,13);    
\coordinate (C2) at (0,9.5);
\coordinate (C3) at (0,5.8);   
\coordinate (C4) at (0,2.9);   
\coordinate (C5) at (0,-0.5);  

\foreach \c in {C1,C2,C3,C4,C5}
    \draw[blossom] (\c) circle (\r);

\node[vertex,label=left:$y$] (y) at ($(C1)+(160:\r)$) {};
\node[base,label={[lab]right:$v_k = b(y)$}] (vk) at ($(C1)+(270:\r)$) {};

\node[vertex] (t2) at ($(C2)+(90:\r)$) {};

\node[base,label={[lab]right:$v_2$}] (v2) at ($(C3)+(270:\r)$) {};

\node[vertex] (t4) at ($(C4)+(90:\r)$) {};
\node[base,label={[lab]right:$v_1$}] (v1) at ($(C4)+(270:\r)$) {};

\node[vertex,label=left:$x$] (x) at ($(C5)+(160:\r)$) {};
\node[vertex] (t5) at ($(C5)+(90:\r)$) {};
\node[base,label={[lab]right:$v_0 = b(x)$}] (v0) at ($(C5)+(270:\r)$) {};

\node[vertex,label={[lab]right:$u_k$}] (uk) at (0,11.25) {};
\draw[squig] (vk)--(uk);
\draw (uk)--(t2);

\draw[dotted,thick] ($(C2)+(270:\r)$)--($(C3)+(90:\r)$);

\node[vertex,label={[lab]right:$u_2$}] (u2) at (0,4.35) {};
\draw[squig] (v2)--(u2);
\draw (u2)--(t4);

\node[vertex,label={[lab]right:$u_1$}] (u1) at (0,1.2) {};
\draw[squig] (v1)--(u1);
\draw (u1)--(t5);

\draw[thick,bend right=35] (y) to (x);

\end{tikzpicture}
\end{center}
\end{subfigure} \quad
\begin{subfigure}{0.3\textwidth}
    \begin{center}
\begin{tikzpicture}[
  every node/.style={circle, draw, inner sep=2pt},
  squig/.style={decorate, decoration={zigzag, segment length=4, amplitude=0.9}, thick},
  scale=0.75, transform shape
]
\node (a) at (1,0) {};
\node (b) at (0,2) {};
\node (c) at (0,4) {};
\node (d) at (0,6) {$y_1$};
\node (e) at (0,8) {$y_2$};
\node (f) at (2,2){$y_3$};
\node (g) at (2,4) {};
\node(h) at (4,2) {};

\node (i) at (5,0) {};
\node (j) at (5,2) {$x_1$};
\node (k) at (5,4) {$x_2$};

\draw[solid,red] (a) -- (b);
\draw[solid] (a) -- (f);
\draw[squig,red] (b) -- (c);
\draw[solid] (c) -- (d);
\draw[squig] (d) -- (e);
\draw[solid,red] (c) -- (g);
\draw[solid,red] (f) -- (h);
\draw[squig,red] (f) -- (g);
\draw[solid] (i) -- (j);
\draw[squig] (j) -- (k);
\end{tikzpicture}
\end{center}
\end{subfigure} \quad
\begin{subfigure}{0.3\textwidth}
   \begin{center}
\begin{tikzpicture}[
  every node/.style={circle, draw, inner sep=2pt},
  squig/.style={decorate, decoration={zigzag, segment length=4, amplitude=0.9}, thick},
  scale=0.75, transform shape
]
\node (x) at (0,0){$x$};
\node (t) at (0,1.5) {$t$}; 
\node (tp) at (0,3) {$t'$};
\node (u) at (0,4.5) {$u$};
\node (up) at (0,6) {$u'$};
\node (s) at (0,7.5){$s$};
\node (sp) at (0,9) {$s'$};
\draw[solid] (x) -- (t);
\draw[squig] (t) -- (tp); 
\draw[solid] (tp) -- (u);
\draw[squig] (u) -- (up);
\draw[solid] (up) -- (s); 
\draw[squig] (s) -- (sp);
\draw[solid] (x) to [out=45, in=315, looseness=1.5] (up);
\draw[solid] (u) to  [ out=45, in=315, looseness=1.5] (sp); 
\end{tikzpicture}
\end{center}
\end{subfigure}
\caption{\label{Illustration of Lemmas}\label{blossom step} Illustration of blossom formation and Lemmas~\ref{related or even-odd} and~\ref{same blossom one}. 
The figure on the left shows the formation of a blossom. The forward edge $xy$ closes a blossom. The path from $b(x)$ to $b(y)$ passes through the odd nodes $u_1$ to $u_k$ and the blossoms with bases $v_1$ to $v_{k}$. The path is readily traced. $u_i$ is the mate of $v_i$, and $v_{i - 1}$ is the base of the blossom containing the parent of $u_i$ for $1 \le i \le k$. 
The figure in the middle illustrates Lemma~\ref{related or even-odd}. The current tree is on the right. The tree on the left is completed and contains an augmenting path shown in red. The node $y_2$ is completely scanned; any other node in the left tree is only partially scanned, either because it belongs to $\CP$ or because it is odd. We may have edges $x_1 y_1$ (case 2), $x_2 y_1$ (case 2), $x_1 y_3$ (case 3), and $y_3 x_2$ (case 4). None of them is added to the search structure. Neither $y_2 x_1$ nor $y_2 x_2$ can exist, since then $x_1$ and $x_2$ or $x_2$ would be added to the first tree. 
The figure on the right illustrates Lemma~\ref{same blossom one}. The addition of the forward edge $xu'$ makes $u$ even; the addition of the edge $us'$ makes $s$ even and forms a blossom containing $t$ and $s$. }
\end{figure}

\begin{lemma}[Lemma A.4 in~\cite{Gabow:GeneralMatching}]\label{same blossom} Let $xy$ be an edge with two even endpoints. If both have been completely scanned, they belong to the same blossom.  \end{lemma}
\begin{proof} Since $x$ and $y$ are completely scanned, they do not belong to $\CP$. Once both nodes are even, $b(x)$ and $b(y)$ are related in the search structure (Lemma~\ref{related or even-odd}). Although, the bases of the blossoms containing $x$ and $y$ may change over time, the ancestor relation does not change. So let us assume that $b(x)$ is an ancestor of $b(y)$. Consider the moment, when $xy$ is scanned from $x$. 

If $y$ is not part of the search structure yet, it would be added as an odd child of $x$. At the time $y$ becomes even, $x$ and $y$ become members of the same blossom. 

So suppose that $y$ is already part of the search structure. Then it is either odd or even, when the edge $xy$ is scanned from $x$. If $y$ is even, a blossom step is executed, after which $x$ and $y$ belong to the same blossom. 

If $y$ is odd, let $t$ be the first odd node on the path from $b(x)$ to $y$, and let $t'$ be its mate. When $x$ scans $xy$, $t'$ is completely scanned. We apply Lemma~\ref{same blossom one} to $t$ and $y$, and conclude that, once $y$ is even, the two nodes belong to the same blossom. Since $t$ was odd, its parent and hence $b(x)$ also belong to the blossom. We have now shown that $x$ and $y$ belong to the same blossom. \end{proof}

\begin{lemma}[Lemma 8.1 in \cite{GT91} and page 10 in \cite{Gabow:GeneralMatching}] The algorithm determines a maximal set of augmenting paths.\end{lemma}
\begin{proof} Assume otherwise. Consider the first time that a call $\findap(x_0)$ for a free node $x_0$ terminates unsuccessfully, although an augmenting path starting in $x_0$ disjoint from the current $\CP$ exists, say $P = x_0, x_1, \ldots, x_{2k + 1}$.  At termination, $x_0$ is part of a blossom with base $x_0$, and $x_{2k+1}$ is part of a blossom with base $x_{2k+1}$. Thus, there must be a non-matching edge $x_i x_{i+1}$ on the path such that $x_i$ is part of a blossom with base $x_j$ with $j \le i$ and $x_{i+1}$ is part of a blossom with base $x_j$ with $j \ge i+1$. We will show that such an edge does not exist.
More precisely, we show by induction on $j$ that every $x_{2j}$ is even, $b(x_{2j}) = x_h$ for some $h \le 2j$, and $x_{2j+1}$ is matched. In particular, $x_{2k + 1}$ is matched, and hence no augmenting path starting in $x_0$ and disjoint from the current $\CP$ exists, a contradiction. 

Every even node $x \not\in \CP$, for which $\findap(x)$ is called, is completely scanned. This holds, since premature termination happens only if an augmenting path is found. 

$x_0$ is a free node not in $\CP$. So $\findap(x_0)$ is called and terminates unsuccessfully. If $x_1$ is free and does not belong to $V(S)$, $\findap(x_0)$ terminates successfully, a contradiction. If $x_1$ is free and already belongs to $V(S)$, it is either an endpoint of a path in $\CP$ (excluded since $P$ is disjoint from $\CP$) or $\findap(x_1)$ was called earlier and terminated unsuccessfully (excluded since $\findap(x_0)$ is the first call that could have found an augmenting path, but did not). This leaves the case that $x_1$ is not free. So the base case $j = 0$ holds. 

We come to the induction step. Assume that the claim holds for $j$. We will establish it for $j + 1$. Since $x_{2j} \not\in \CP$, it has been completely scanned. Also, $x_{2j+1}$ is matched by the IH. Let $x_{2j + 2}$ be the mate of $x_{2j+1}$. So, $\findap(x_{2j + 2})$ is called and terminates unsuccessfully.  In this call, the edge to $x_{2j + 3}$ is inspected. Assume first that $x_{2j + 3}$ is free. If it does not belong to $V(S)$ yet, an augmenting path is found, a contradiction. If it already belongs to $V(S)$ and since it does not belong to $\CP$, $\findap(x_{2j + 3})$ terminated unsuccessfully before $\findap(x_0)$ was called, a contradiction to our assumption that $\findap(x_0)$ is the first such call. Thus, $x_{2j + 3}$ is not free. 

We now distinguish cases, according to the parity of $x_{2j+1}$ at termination of the call $\findap(x_0)$. 


If $x_{2j + 1}$ is odd, $x_{2j+1}$ is a singleton blossom, and $x_{2j + 2}$ is matched to a vertex outside the blossom containing it. So, $x_{2j + 2}$ is the base of the blossom. Hence $b(x_{2j + 2}) = x_{2j + 2}$, and the induction step is completed.

If $x_{2j + 1}$ is even, it was completely scanned and hence belongs to the same blossom as $x_{2j}$ by Lemma~\ref{same blossom}. Thus $b(x_{2j+1}) = b(x_{2j})$ and hence $b(x_{2j + 1})= x_h$ for some $h \le 2j$ by induction hypothesis. In particular, $b(x_{2j + 1}) \not= x_{2j + 1}$. So $x_{2j+1}$ is not the base of a blossom, and hence the matching edge connecting it to $x_{2j + 2}$ connects two nodes in the same blossom. Thus $x_{2j + 2}$ is even and $b(x_{2j + 2}) = b(x_{2j + 1}) = x_h$.  \end{proof}

We have now described how to find a maximal number of edge-disjoint augmenting paths in $H$. It remains to lift these paths to $G$ by filling in the parts inside maximal blossoms. This is standard. Suppose, we need to find the path from a vertex $z$ in a blossom to the base $b$ of the blossom. If $z$ was born even, we can walk straight down. If $z$ was born odd, we need to go through the bridge of the blossom. For this purpose, we store with each node of a blossom a pointer to the bridge, say $xy$, and also which node of the bridge is on the $z$-side of the bridge, say $x$. Then we construct the path inside the blossom by walking down from $x$ to $z$ and from $y$ to $b$. We collect the non-matching edges of the path in a set. The non-matching edges of the path are returned as a list of edges. Of course, if $z$ was born odd, the bridge is also added to this set. 

Once we have lifted all augmenting paths to $G$, augmenting them is easy. We simply mate the endpoints of each non-matching edge in an augmenting path. This automatically breaks the old partnerships. \smallskip

\textbf{Remark:} In the implementation, we never construct $H$ explicitly, but use $G$ instead. The bases of the maximal blossoms represent the nodes of $H$.

\section{Running Time Analysis}

The algorithm works iteratively. In each iteration, it augments a maximal number of edge-disjoint augmenting paths. Therefore, the number of iterations is $O(\sqrt{n})$~\cite{Hopcroft-Karp}. Moreover, it is well-known how to make each iteration run in time $O(m \alpha(n))$ or even $O(m)$~\cite{Gabow-Tarjan:union-find, GT91,Gabow:GeneralMatching}.  We give a short account for completeness. Part I explores the graph breadth-first; Part III explores it depth-first. The non-trivial actions in Part I are the discovery and the administration of blossoms; in Part III, it is only the administration of blossoms.  In part I, whenever a bridge $uv$ is discovered, one walks \emph{in lock-step fashion} from $u$ and $v$ towards the roots of the search structures containing them until either the same node is encountered (the base of the newly formed blossom) on both paths or distinct roots are reached. In the first case, a new blossom has been found in time proportional to the size of the blossom and hence in time proportional to the size reduction of the graph, and in the second case, an augmenting path has been found. Thus, the time for finding blossoms is $O(n)$. A union-find data structure is used for maintaining blossoms; the cost is $O(m \alpha(n))$ or $O(m)$ depending on the sophistication of the realization.

\section{Implementation and Experimentation}\label{Implementation}
The authors have implemented the algorithm in $\CC$. The source code is available on the companion webpage {\scriptsize \url{https://people.mpi-inf.mpg.de/~mehlhorn/CompanionPageGenMatchingImplementation.html}}. The companion webpage also contains an implementation of Gabow's algorithm. We ran both implementations on many examples. The two implementations exhibit about the same running times. The paper~\cite{GabowImplementation} contains detailed data for Gabow's algorithm. 

Table~\ref{instruction counts} shows the results of some experiments. We used random graphs and carefully constructed graphs. The latter graphs come in two kinds: short chains only or short and long chains. Both kinds consist of a complete graph with about $\sqrt{n}$ vertices, $O(n)$ short chains of length seven all attached to a fixed vertex $z$ of the complete graph, and, in the case of short and long chains, one chain of length $2i + 1$ each attached to $z$, where $4 \le i \le \sqrt{n}$. We refer to~\cite{GabowImplementation} for more details. On the graphs with short and long chains, the algorithm executes $O(\sqrt{n})$ iterations; on the graphs with only short chains, the algorithm executes a constant number of iterations; and on sparse random graphs, the algorithm executes a non-linear number of iterations. In random graphs with sufficient edge density ($m/n \ge 400$), non-maximum matchings have logarithmic length augmenting paths with high probability~\cite{MatchingSparseRandomGraphs}.

\begin{table}[ht]
\begin{center}
\begin{tabular}{||r|r||r|r|r||r|r||r|r||}\hline
& & \multicolumn{3}{c||}{SHORT AND LONG} & \multicolumn{2}{c||}{SHORT} & \multicolumn{2}{c||}{RANDOM} \\ \hline
$\sfrac{n}{10^3}$ &  $\sfrac{m}{10^3}$   & \#it &  $\sfrac{I(n)}{10^6}$ & $\sfrac{I(n)}{I(n/2)}$ & $\sfrac{I(n)}{10^6}$ & $\sfrac{I(n)}{I(n/2)}$  & $\sfrac{I(n)}{10^6}$ & $\sfrac{I(n)}{n \ln n}$ \\ \hline
10 & 22 & 24 & 91.8 & & & & & \\ \hline
20& 56 & 33 & 264.6 & 2.88  & & & 80.1 & 4.0\\ \hline
40 & 114  & 47  & 744.4 & 2.81  & 7.5 &  &  174.0 & 4.1\\ \hline
80 & 232& & & & 15.1& 2.0 &  426.7 & 4.73 \\ \hline
160 & 471 & & & & 30.1 & 2.0 & &   \\  \hline
\end{tabular}\end{center}
\caption{\textbf{Instruction Counts, Short and Long Chains, Short Chains, Random Graphs:}\label{instruction counts} On graphs with short and long chains, the algorithm executes $O(\sqrt{n})$ iterations. This is confirmed by the third column; the number of iterations doubles when $n$ is quadrupled. The fourth column shows the counts for the number of accesses to edge records, an instruction that is representative for the execution of the algorithm. The instruction counts grow as $n^{3/2}$; note that $2^{3/2} \approx 2.83$. \\ \protect
On graphs with only short chains, the number of iterations is constant, and the instruction count grows linearly. \\ \protect
On random graphs, the number of iterations grows superlinearly. It is known~\cite{MatchingSparseRandomGraphs} that for $m/n \ge 400$, non-maximum matchings in random graphs have an augmenting path of length $O(\log n)$ with high probability. \\ \protect
The programs were compiled with the flag -pg, and the profiler gprof was used for determining instruction counts. 
}\end{table}

\section{The Connection to Gabow's Algorithm}\label{Connection to Gabow}

Part I of Gabow's algorithm determines the length of the shortest augmenting paths, if any, and collects enough information for the construction of an auxiliary graph with the property that its augmenting paths are in one-to-one correspondence with the $\saps$. The algorithm is a primal-dual algorithm that is based on Edmonds' algorithm for maximum weight matching. 

Let $M$ be the current matching. Each edge of $G$ is given a weight: weight two for $e \in M$ and weight zero for $e \not\in M$. The weight of an augmenting path $p$ is defined as 
\[      w(P) = w(P \setminus M) - w(P \cap M),\]
i.e., as the increase in weight obtained by augmenting $p$ to $M$. The increase is the negative of the number of edges in $P \cap M$ and is maximum if $\abs{P \cap M}$ is smallest, i.e., if $p$ is a $\sap$.  Edmonds' algorithm is based on the linear programming formulation of maximum weighted matching and makes use of linear programming duality. The dual linear program has a variable for each vertex of $G$ and for each set of vertices of odd cardinality three or more. The latter are non-negative; the former are unconstrained. The \emph{reduced weight} $\hat{w}(e)$ of an edge $e = xy$ is defined as
\[ \hat{w}(e) = d(x) + d(y) + \sum_{x,y \in B} z(B) - w(e), \]
where $d$ denotes dual values of vertices and $z$ denotes dual values of odd sets (of cardinality three or more). The sum is over all odd sets containing $x$ and $y$. \emph{Reduced weights are always non-negative, and an edge is called tight if its reduced weight is zero}.
When the search for an augmenting path starts, $d(v) = 1$ for all $v$ and $z(B) = 0$ for all odd sets $B$. Then matching edges are tight and non-matching edges are non-tight. The search for an augmenting path grows search structures (usually called trees) rooted at free nodes. The trees are grown concurrently, and the endpoints of a matching edge either both belong to a tree or none belongs to a tree. The trees are initialized with the free nodes. A tree node is even, if there is an even length path in the search structure connecting it to its root, and odd otherwise. The incoming tree-edges of even nodes are matching edges and are non-matching for odd nodes. Trees are grown by the addition of \emph{tight} edges incident to even nodes. So assume that $x$ is an even node and let $xy$ be a tight edge incident to $x$. If $y$ does not belong to any search structure yet (then $y$ is matched), we add $y$ and its mate to the tree containing $x$; $y$ becomes odd, and its mate becomes even. When $y$ already belongs to a search structure and is odd, we do nothing, as we have simply discovered another odd length path to $y$. If $y$ is even and belongs to a different tree, we have discovered an augmenting path. If $y$ is even and belongs to the same tree, we have discovered a so-called blossom. Let $b$ be the lowest common ancestor of $x$ and $y$. Then all odd nodes on the paths from $b$ to $x$ and $y$ become even. For example, an odd node $z$ between $b$ and $y$ can now be reached by going from the root to $x$ (even length), then to $y$ (one step), and then from $y$ towards $z$ (odd length). The edge $xy$ is called the \emph{bridge} of the blossom, and the blossom consists of all nodes on the paths from $b$ to $x$ and $y$, respectively. A blossom contains an odd number of nodes. We contract all nodes of the blossom into a single node. In this way, the search structures stay trees. Note that blossoms can be nested.

In what order are edges added to the search structures? We maintain the invariant that all edges in the search structures are tight and that all roots have the same dual value. We grow the trees by adding non-matching tight edges, i.e., edges $e = xy$ with $\hat{w}(e) = 0$. Then, necessarily, $d(x) + d(y) = 0$ as $w(e) = 0$. In particular, if $y$ does not belong to a search structure yet, $d(x) = -1$ and $d(y) = +1$. Suppose now that we cannot further grow trees, i.e., $\hat{w}(e) > 0$ for any edge $xy$ with at least one even endpoint. Then we perform a dual update. We decrease $d(v)$ by one for every even vertex, increase $d(v)$ by one for every odd vertex, and increase $z(B)$ by two for every maximal blossom, i.e., any blossom not contained in any other blossom. Then, $z$-values are always even, and the $d$-values of all vertices in the search structures have the same parity. Moreover, all edges in the search structures stay tight: if one endpoint is even and one is odd, this is obvious; if both endpoints are even, the endpoints are contained in the same maximal blossom, and again the claim is obvious. Consider now an edge $xy$ not belonging to any search structure. If no endpoint belongs to the search structures, $d(x) = d(y) = 1$ and the reduced weight does not change. If at least one endpoint is odd, the reduced weight does not decrease. If one endpoint is even, and the other endpoint does not belong to a search structure, the reduced weight decreases by one and hence stays non-negative. If both endpoints are even and do not belong to the same blossom, the reduced weight decreases by two and hence stays non-negative since the $d$-values of the endpoints have the same parity. In this way, reduced weights stay non-negative, edges in the search structures stay tight, and all roots have the same dual value. After $\Delta$ dual updates, we have $d(f) = 1 - \Delta$ for any free node $f$.
Once a tight edge having endpoints in different trees is added to the search structure, Part I ends and a sap is found. Its length is $2 \Delta - 1$. 

The connection to our version of part one is simple. At all times:
\begin{align*} 
   \lcp(v) &= d(v) - d(f) = d(v) + \Delta - 1   && \text{if $v$ is even,} \\
   \lcpodd(v) &= -d(v) - d(f) + 1 = -d(v) + \Delta &&\text{if $v$ is odd,}
   \end{align*}
   where $f$ is any free node. Gabow states in his paper, that the lengths of the canonical paths satisfy these equations. We first rewrote the algorithm using $\lcp$-values and $\lcpodd$-values instead of the dual values $d$ and then proved correctness without any reference to linear programming duality. We stress that the algorithms are isomorphic in the sense that they can perform exactly the same sequence of growth- and bridge-steps. However, we feel that our presentation is more direct.

\end{document}